\documentclass[11pt]{article}
\usepackage{amsmath,amssymb,amsthm,xspace}
\usepackage{tikz,tikz-qtree}
\usepackage{expdlist}
\usepackage{authblk}
\usepackage{fullpage}
\usepackage{comment}
\usepackage{enumerate}
\usepackage{algorithm,algorithmicx,algpseudocode, enumerate,amssymb}
 \usepackage[hidelinks]{hyperref}
 \usepackage{cleveref}
 \usepackage[utf8]{inputenc}

 \usetikzlibrary{shapes,arrows}

\usetikzlibrary{shapes.multipart}
\usetikzlibrary{decorations.text}
 
 \newtheorem{theorem}{Theorem}[section]
 \newtheorem{claim}[theorem]{Claim}
\newtheorem{corollary}[theorem]{Corollary}
\newtheorem{lemma}[theorem]{Lemma}
\newtheorem{fact}[theorem]{Fact}

\newtheorem{observation}[theorem]{Observation}

\makeatletter
\newtheorem*{rep@theorem}{\rep@title}
\newcommand{\newreptheorem}[2]{%
\newenvironment{rep#1}[1]{%
 \def\rep@title{#2 \ref{##1}}%
 \begin{rep@theorem}}%
 {\end{rep@theorem}}}
\makeatother
\newreptheorem{theorem}{Theorem}
\newreptheorem{corollary}{Corollary}
\newreptheorem{lemma}{Lemma}

\DeclareMathOperator{\lcp}{lcp}

\DeclareMathOperator{\polylog}{polylog}
\DeclareMathOperator{\polyloglog}{polyloglog}

\newcommand{\Oh}{\mathcal{O}}
\newcommand{\floor}[1]{\left\lfloor #1 \right\rfloor}
\newcommand{\eps}{\varepsilon}
\newcommand{\sub}{\subseteq}

\newcommand{\rank}{\mathrm{rank}}
\newcommand{\select}{\mathrm{select}}
\newcommand{\parent}{\mathrm{parent}}
\newcommand{\rchild}{\mathrm{rchild}}
\newcommand{\lchild}{\mathrm{lchild}}
\renewcommand{\S}{\mathcal{S}}

\newcommand{\mayqed}{}

\begin{document}
\title{Wavelet Trees Meet Suffix Trees\thanks{The third author is supported by Polish budget funds for science in 2013-2017 as a research project under the `Diamond Grant' program. The fourth author is partly supported by Dynasty Foundation.}}

\author[1]{Maxim Babenko}
\author[2]{Paweł Gawrychowski}
\author[3]{Tomasz Kociumaka}
\author[1]{Tatiana Starikovskaya}

\affil[1]{National Research University Higher School of Economics (HSE)}
\affil[ ]{\texttt{mbabenko@hse.ru, tstarikovskaya@hse.ru}}
\affil[2]{Max-Planck-Institut f\"{u}r Informatik}
\affil[ ]{\texttt{gawry@cs.uni.wroc.pl}}
\affil[3]{Institute of Informatics, University of Warsaw}
\affil[ ]{\texttt{kociumaka@mimuw.edu.pl}}

\date{\empty}
\maketitle

\begin{abstract}
We present an improved wavelet tree construction algorithm
and discuss its applications to a number of rank/select problems for
integer keys and strings.

Given a string of length $n$ over an alphabet of size $\sigma\leq n$,
our method builds the wavelet tree in $\Oh(n \log \sigma/ \sqrt{\log{n}})$ time,
improving upon the state-of-the-art algorithm by a factor of $\sqrt{\log n}$.
As a consequence, given an array of $n$ integers we can construct in $\Oh(n \sqrt{\log n})$ time
a data structure consisting of $\Oh(n)$ machine words and capable of answering rank/select queries
for the subranges of the array in $\Oh(\log n / \log \log n)$ time. This is
a $\log \log n$-factor improvement in query time compared to Chan and P\u{a}tra\c{s}cu (SODA 2010)
and a $\sqrt{\log n}$-factor improvement in construction time compared to Brodal et al. (Theor. Comput. Sci. 2011).
We also design an offline algorithm for range successor queries and improve construction time for $o(\log n)$-time range successor queries in the online setting.

Next, we switch to stringological context and propose a novel notion of \emph{wavelet suffix trees}.
For a string $w$ of length $n$, this data structure occupies $\Oh(n)$ words, takes $\Oh(n \sqrt{\log n})$ time
to construct, and simultaneously captures the combinatorial structure of substrings of $w$
while enabling efficient top-down traversal and binary search. In particular, with a wavelet suffix tree we can answer in $\Oh(\log |x|)$ time two analogues of rank/select queries for suffixes of substrings:
\begin{enumerate}[1)] \compact
    \item For substrings $x$ and $y$ of $w$
    (given by their endpoints) count the number of suffixes of $x$ that
    are lexicographically smaller than $y$;

    \item For a substring $x$ of $w$ (given by its endpoints) and an integer $k$,
    find the $k$-th lexicographically smallest suffix of $x$.
\end{enumerate}
We further show that wavelet suffix trees allow to compute a run-length-encoded Burrows-Wheeler transform of a substring $x$ of $w$ (again, given by its endpoints) in $\Oh(s \log |x|)$ time,
where $s$ denotes the length of the resulting run-length encoding. This answers a question
by Cormode and Muthukrishnan (SODA 2005), who considered an analogous problem for Lempel-Ziv compression.

All our algorithms, except for the construction of wavelet suffix trees,
which additionally requires $\Oh(n)$ time in expectation, are deterministic and operate in the word RAM model.

\end{abstract}

\section{Introduction}
Let $\Sigma$ be a finite ordered non-empty set which we refer to as an \emph{alphabet}. The elements of~$\Sigma$ are called \emph{characters}. characters are treated as integers in a range $[0, \sigma-1]$; and we assume that a pair of characters can be compared in $\Oh(1)$ time. A finite ordered sequence of characters (possibly empty) is called a \emph{string}.
characters in a string are enumerated starting from~$1$, that is, a string $w$ of \emph{length} $n$ consists of characters $w[1], w[2], \ldots, w[n]$. 

\paragraph{Wavelet trees.}
The wavelet tree, invented by Grossi, Gupta, and Vitter \cite{Grossi:2003:HET:644108.644250}, is an important data structure with a vast number of applications to stringology, computational geometry, and others (see \cite{DBLP:journals/jda/Navarro14} for an excellent survey). Despite this, the problem of wavelet tree construction has not received much attention in the literature. For a string $w$ of length $n$, one can derive a construction algorithm with $\Oh(n \log \sigma)$ running time directly from the definition. Apart from this, two recent works~\cite{WaveletTreeConstruction1,WaveletTreeConstruction2} present construction algorithms in the setting when limited extra space is allowed. Running time of the algorithms is higher than that of the naive algorithm without these space restrictions.

The first result of our paper is a novel deterministic algorithm for constructing a wavelet tree in $\Oh(n \log \sigma/ \sqrt{\log{n}})$ time. 
Our algorithm is capable of building arbitrary-shaped (binary) wavelet trees of $\Oh(\log \sigma)$ height in asymptotically the same time.
Hence, it can be used to construct wavelet trees which do not form a perfect binary tree, for example wavelet trees for Huffman encoding~\cite{DBLP:conf/soda/GrossiGV04}, or wavelet tries~\cite{DBLP:conf/pods/GrossiO12} that have the shape of a trie for a given set of strings.
Similar results on wavelet trees were obtained independently by Munro, Nekrich, and Vitter~\cite{DBLP:conf/spire/MunroNV14}.

Our contribution also lies in applying the wavelet tree construction algorithm and the underlying ideas to derive some improved bounds for range queries. Namely, given an array $A[1..n]$ of integers, one could ask to compute the number of integers in $A[i..j]$ that
are smaller than a given $A[k]$ (\emph{range rank} query);
or to find, for given $i$, $j$, and~$k$,
the $k$-th smallest integer in $A[i..j]$ (\emph{range select} query);
or to determine the smallest value in $A[i..j]$ which is at least as large as a given $A[k]$ (\emph{range successor} query).

These problems have been widely studied; see e.g.
Chan and P\u{a}tra\c{s}cu~\cite{DBLP:conf/soda/ChanP10} and
Brodal et al.~\cite{BrodalMedian} for range rank/select queries and \cite{SortedRangeReporting,DBLP:journals/corr/Zhou13b} for range successor queries.
By slightly tweaking our construction of wavelet trees, we can build in $\Oh(n \sqrt{\log n})$ deterministic time
an $\Oh(n)$ size data structure for answering range rank/select queries in $\Oh(\frac{\log n}{\log \log n})$
time. 
Our approach yields
a $\sqrt{\log n}$-factor improvement to the construction time upon
Brodal et al. \cite{BrodalMedian} and
a $\log \log n$-factor improvement to the query time upon
Chan and P\u{a}tra\c{s}cu \cite{DBLP:conf/soda/ChanP10}.
To answer range successor queries online, we show a reduction to range rank and select queries, which yields the same bounds.
Additionally, we show how to answer $q$ range successor queries in an offline manner
in $\Oh((n+q)\sqrt{\log n})$ total time using a wavelet-tree-based algorithm.
In particular, this implies $\Oh(n\sqrt{\log n})$ time for $q=\Theta(n)$ queries, which improves by a factor of $\sqrt{\log n}$ over previously known solutions.
To the best of our knowledge, no prior algorithm could answer $q=\Theta(n)$ queries in $o(n\log n)$ time.
The previously existing solutions have better query time, but require $\Omega(n\log n)$-time
preprocessing. More precesiely, $\Oh(\log\log n)$ query time can be achieved using $\Oh(n\log\log n)$ space~\cite{DBLP:journals/corr/Zhou13b}
and $\Oh(\log^{\eps} n)$ query time (for any $\eps>0$) with $\Oh(n)$ space~\cite{SortedRangeReporting}.

\paragraph{Wavelet suffix trees.}
Then we switch to stringological context and extend our approach to the so-called
\emph{internal} string problems (see~\cite{DBLP:journals/corr/KociumakaRRW13}).
This type of problems involves construction of a compact data structure for a given fixed string $w$
capable of answering certain queries for substrings of $w$. This line of development was originally inspired by suffix trees,
which can answer some basic internal string queries (e.g. equality testing and longest common extension
computation) in constant time and linear space.
Lately a number of studies emerged addressing compressibility \cite{CormodeMuthukrishnan05,Keller201442},
range longest common prefixes (range LCP) \cite{RangeLCP,FasterRangeLCP},
periodicity \cite{FactorPeriodicity2012,Crochemore:2010:EPP:1928328.1928362},
minimal/maximal suffixes \cite{Minmaxsuf,MinmaxsufRevisited}, substring hashing~\cite{PerfectHashing,WeightedAncestors}, and fragmented pattern matching~\cite{FragmentedPatternMatching,WeightedAncestors}.

Our work focuses on rank/select problems for suffixes of a given substring.
Given a fixed string~$w$ of length $n$, the goal is to preprocess it into a compact data structure
for answering the following two types of queries efficiently:
\begin{enumerate}[1)]
    \item \emph{Substring suffix rank}: For substrings $x$ and $y$ of $w$
    (given by their endpoints) count the number of suffixes of $x$ that
    are lexicographically smaller than~$y$;

    \item \emph{Substring suffix select}: For a substring $x$ of $w$ (given by its endpoints) and an integer $k$,
    find the $k$-th lexicographically smallest suffix of $x$.
\end{enumerate}

Note that for $k = 1$ and $k = |x|$ substring suffix select queries reduce to computing
the lexicographically minimal and the lexicographically maximal suffixes of a given substring.
Study of this problem was started by Duval~\cite{Duval}. He showed that the maximal and the minimal suffixes of \emph{all prefixes} of a string can be found in linear time and constant additional space. Later this problem was addressed in~\cite{Minmaxsuf,MinmaxsufRevisited}. In~\cite{MinmaxsufRevisited} it was shown that the minimal suffix of any substring can be computed in $\Oh(\tau)$ time by a linear space data structure with $\Oh(n \log n / \tau)$ construction time for any parameter $\tau$, $1 \le \tau \le \log n$. As an application of this result it was shown how to construct the Lyndon decomposition of any substring in $\Oh(\tau s)$ time, where $s$ stands for the number of distinct factors in the decomposition. 
For the maximal suffixes an optimal linear-space data structure with $\Oh(1)$ query and $\Oh(n)$ construction time was presented. We also note that~\cite{MakinenRankSelect} considered a problem with a similar name, namely substring rank and select. However, the goal there is to preprocess a string, so that given any pattern, we can count its occurrences in a specified prefix of the string, and select the $k$-th occurrence in the whole string. One can easily see that this problem is substantially different than the one we are interested in.

Here, we both generalize the problem to an arbitrary $k$ (thus enabling to answer general \emph{substring suffix select} queries) and also consider \emph{substring suffix rank} queries. We devise a linear-space data structure with $\Oh(n \sqrt{\log n})$ expected construction time supporting both types of the queries in $\Oh(\log |x|)$ time.

Our approach to substring suffix rank/select problems is based on wavelet trees and attracts a number of additional combinatorial and algorithmic ideas. Combining wavelet trees with suffix trees we introduce a novel concept of \emph{wavelet suffix trees}, which forms a foundation of our data structure. Like usual suffix trees, wavelet suffixes trees provide a compact encoding for all substrings of a given text; like wavelet trees they maintain logarithmic height. 
Our hope is that these properties will make wavelet suffix trees an attractive tool for a large variety of stringological problems.

We conclude with an application of wavelet suffix trees to substring compression,
a class of problems introduced by Cormode and Muthukrishnan~\cite{CormodeMuthukrishnan05}.
Queries of this type ask for a compressed representation of a given substring.
The original paper, as well as a more recent work by Keller et al.~\cite{Keller201442},
considered Lempel-Ziv compression schemes. Another family of methods, based on Burrows-Wheeler transform~\cite{BWT}
was left open for further research. We show that wavelet suffix trees allow to compute a run-length-encoded Burrows-Wheeler transform of an arbitrary substring $x$ of $w$ (again, given by its endpoints) in $\Oh(s \log |x|)$ time, where $s$ denotes the length of the resulting run-length encoding.

\section{Wavelet Trees}\label{sec:wt}
Given a string $s$ of length $n$ over an alphabet $\Sigma = [0,\sigma-1]$, where $\sigma\leq n$,
we define the wavelet tree of $s$ as follows. First, we create the root node $r$ and construct its
bitmask $B_r$ of length $n$. To build the bitmask, we think of every
$s[i]$ as of a binary number consisting of exactly $\log\sigma$ bits (to make the
description easier to follow, we assume that $\sigma$ is a power of $2$), and put
the most significant bit of $s[i]$ in $B_{r}[i]$. Then we partition $s$ into $s_{0}$ and $s_{1}$
by scanning through $s$ and appending $s[i]$ with the most significant bit removed to
either $s_{1}$ or $s_{0}$, depending on whether the removed bit
of $s[i]$ was set or not, respectively. Finally, we recursively define the wavelet trees
for $s_{0}$ and $s_{1}$, which are strings over the alphabet $[0,\sigma/2-1]$,
and attach these trees to the root. We stop when the alphabet is unary. The final result
is a perfect binary tree on $\sigma$ leaves with a bitmask attached at every non-leaf node.
Assuming that the edges are labelled by {\bf 0} or {\bf 1} depending
on whether they go to the left or to the right respectively, we can define a \emph{label} of a node 
to be the concatenation of the labels of the edges on the path from the root to this node.
Then leaf labels are the binary representations of characters in $[0,\sigma-1]$; see Figure~\ref{fig:wavelet example1} for an example.
 
In virtually all applications,
each bitmask $B_{r}$ is augmented with a rank/select structure. For a bitmask $B[1..N]$ this structure implements operations $\rank_{b}(i)$, which counts
the occurrences of a bit $b \in \{0,1\}$ in $B[1..i]$, and $\select_{b}(i)$, which selects the
$i$-th occurrence of $b$ in the whole $B[1..n]$, for any $b\in\{0,1\}$,
both in $\Oh(1)$ time.

\begin{figure}[tbp]
\begin{center}
\begin{tikzpicture}[scale=0.88, every node/.style={transform shape},level/.style={level distance=10mm,sibling distance=48mm/#1},level 2/.style={sibling distance=24mm},level 3/.style={sibling distance=12mm}, level 4/.style={sibling distance=6mm},minimum width=15pt]
\node [draw] {1011100000101011}
child {
 node [draw] {11100001}
 child {
  node [draw] {0011}
  child {
   node [draw] {01}
   child {
    node [ellipse,draw,inner sep=0pt] {$0$}
   }
   child {
    node [ellipse,draw,inner sep=0pt] {$1$}
   }
  }
  child {
   node [draw] {01}
   child {
    node [ellipse,draw,inner sep=0pt] {$2$}
   }
   child {
    node [ellipse,draw,inner sep=0pt] {$3$}
   }
  }
 }
 child {
  node [draw] {1100}
  child {
   node [draw] {01}
   child {
    node [ellipse,draw,inner sep=0pt] {$4$}
   }
   child {
    node [ellipse,draw,inner sep=0pt] {$5$}
   }
  }
  child {
   node [draw] {10}
   child {
    node [ellipse,draw,inner sep=0pt] {$6$}
   }
   child {
    node [ellipse,draw,inner sep=0pt] {$7$}
   }
  }
 }
}
child {
 node [draw] {10100101}
 child {
  node [draw] {1010}
  child {
   node [draw] {10}
   child {
    node [ellipse,draw,inner sep=0pt] {$8$}
   }
   child {
    node [ellipse,draw,inner sep=0pt] {$9$}
   }
  }
  child {
   node [draw] {10}
   child {
    node [ellipse,draw,inner sep=0pt] {$10$}
   }
   child {
    node [ellipse,draw,inner sep=0pt] {$11$}
   }
  }
 }
 child {
  node [draw] {0101}
  child {
   node [draw] {01}
   child {
    node [ellipse,draw,inner sep=0pt] {$12$}
   }
   child {
    node [ellipse,draw,inner sep=0pt] {$13$}
   }
  }
  child {
   node [draw] {10}
   child {
    node [ellipse,draw,inner sep=0pt] {$14$}
   }
   child {
    node [ellipse,draw,inner sep=0pt] {$15$}
   }
  }
 }
};
\end{tikzpicture}
\end{center}
\caption{Wavelet tree for a string $1100_2$ $0111_2$  $1011_2$  $1111_2$  $1001_2$  $0110_2$  $0100_2$  $0000_2$  $0001_2$  $0010_2$  $1010_2$  $0011_2$  $1101_2$  $0101_2$  $1000_2$  $1110_2$, the leaves are labelled with their corresponding characters.}
\label{fig:wavelet example1}
\end{figure}
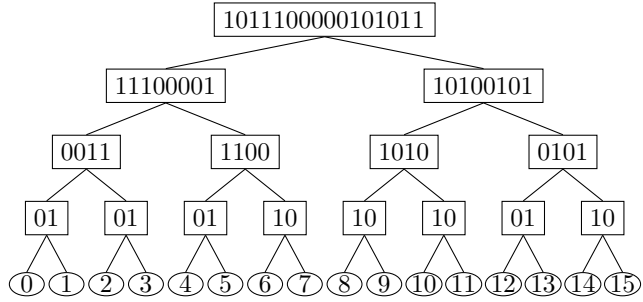

The bitmasks and their corresponding rank/select structures are stored one after another, each
starting at a new machine word. The total space occupied by the bitmasks alone
is $\Oh(n \log \sigma)$ bits, because there are $\log\sigma$ levels, and the lengths of
the bitmasks for all nodes at one level sum up to $n$. A rank/select structure built for a bit
string $B[1..N]$ takes $o(N)$ bits~\cite{Clark,Jacobson}, so the space taken by all of them is
$o(n\log \sigma)$ bits. Additionally, we might lose one machine word per node
because of the word alignment, which sums up to $\Oh(\sigma)$.
For efficient navigation, we number the nodes in a heap-like fashion and, using $\Oh(\sigma)$ space,
store for every node the offset where its bitmasks and the corresponding rank/select
structure begin. Thus, the total size of a wavelet tree
is $\Oh(\sigma + n \log \sigma / \log n)$, which for $\sigma \le n$ is
$\Oh(n \log \sigma / \log n)$.

Directly from the recursive definition, we get a construction algorithm taking $\Oh(n\log\sigma)$ time,
but we can hope to achieve a better time complexity as the size of the output is just $\Oh(n \log \sigma / \log n)$
words. In Section~\ref{sec:WTconstruction} we show that one can construct all bitmasks augmented
with the rank/select structures in total time $\Oh(n \log \sigma/ \sqrt{\log{n}})$.

\paragraph{Arbitrarily-shaped wavelet trees.}
A generalized view on a wavelet tree is that apart from a string $s$ we are given
an arbitrary full binary tree $T$ (i.e., a rooted tree whose nodes have 0 or 2 children) on $\sigma$ leaves, together 
with a bijective mapping between the leaves and the characters in $\Sigma$.
Then, while defining the bitmasks~$B_v$, we do not remove the most significant bit of each character,
and instead of partitioning the values based on this bit,
we make a decision based on whether the leaf corresponding to the character lies in the left or in the right subtree of $v$.
Our construction algorithm generalizes to such arbitrarily-shaped wavelet trees without increasing the time complexity
provided that the height of $T$ is $\Oh(\log\sigma)$.

\paragraph{Wavelet trees of larger degree.}
Binary wavelet trees can be generalized to higher degree trees in a natural way as follows.
We think of every $s[i]$ as of a number in base $d$. The degree-$d$ wavelet
tree of $s$ is a tree on $\sigma$ leaves, where every inner node is of degree $d$, except for the
last level, where the nodes may have smaller degrees. First, we create
its root node $r$ and construct its string $D_{r}$ of length $n$ setting as $D_r[i]$ the most significant digit of $s[i]$. We partition
$s$ into $d$ strings $s_{0},s_{1},\ldots,s_{d-1}$ by scanning through $s$ and appending
$s[i]$ with the most significant digit removed to $s_{a}$, where $a$ is the removed
digit of $s[i]$. We recursively repeat the construction for every $s_{a}$ and attach
the resulting tree as the $a$-th child of the root. All strings $D_{u}$ take
$\Oh(n\log \sigma)$ bits in total, and every $D_{u}$ is augmented with a generalized rank/select
structure.

We consider $d=\log^{\eps}n$, for a small constant $\eps>0$. Remember that we assume
$\sigma$ to be a power of $2$, and because of similar technical reasons we assume $d$
to be a power of two as well. Under such assumptions, our construction algorithm for
binary wavelet trees can be used to construct a higher degree wavelet tree. More precisely,
in Section~\ref{sec:higher degree construction} we show how to construct such
a higher degree tree in $\Oh(n\log\log n)$, assuming that we are already given the binary
wavelet tree, making the total construction time $\Oh(n\sqrt{\log n})$ if $\sigma\leq n$, and allowing
us to derive improved bounds on range selection, as explained in Section~\ref{sec:higher degree construction}.

\subsection{Wavelet Tree Construction}
\label{sec:WTconstruction}
Given a string $s$ of length $n$ over an alphabet $\Sigma$, we want to construct its wavelet tree, which 
requires constructing the bitmasks and their rank/select structures, in $\Oh(n \log \sigma / \sqrt{\log n})$ time.

\paragraph{Input representation.}
The input string $s$ can be considered as a list of $\log \sigma$-bit integers. We assume that it is represented in a packed form, as described below.

A single machine word of length $\log n$ can accommodate $\frac{\log n}{b}$ $b$-bit integers.
Therefore a list of $N$ such integers can be stored in $\frac{Nb}{\log n}$ machine words. (For $s$, $b = \log \sigma$, but later we will also consider other values of $b$.)
We call such a representation a \emph{packed list}.

We assume that packed lists are implemented as resizable bitmasks (with each block of $b$ bits
representing a single entry) equipped with the size of the packed list.
This way a packed list of length $N$ can be appended to another packed list in $\Oh(1+\frac{Nb}{\log n})$ time,
since our model supports bit-shifts in $\Oh(1)$ time.
Similarly, splitting into lists of length at most $k$ takes
$\Oh(\frac{Nb}{\log n}+\frac{N}{k})$ time.

\paragraph{Overview.}
Let $\tau$ be a parameter
to be chosen later to minimize the total running time.
We call a node $u$ \emph{big} if its depth is a multiple of $\tau$, and \emph{small} otherwise. 
The construction conceptually proceeds in two phases.

First, for every big node $u$ we build a list $S_u$. Remember that the label $\ell_{u}$ of a node $u$ at distance $\alpha$
from the root consists of $\alpha$ bits, and the characters corresponding to the leaves below $u$ are exactly the
characters whose  binary representation starts with $\ell_u$. $S_u$ is defined to
be a subsequence of $s$ consisting of these characters. 

Secondly, we construct the bitmasks $B_{v}$ for every node $v$ using $S_u$ of the nearest big ancestor $u$ of $v$ (possibly $v$ itself).

\paragraph{First phase.} Assume that for a big node $u$ at distance $\alpha \tau$ from the root we are given the list~$S_u$.
We want to construct the lists $S_v$ for all big nodes $v$ whose deepest (proper) big ancestor is~$u$.
 There are exactly $2^{\tau}$ such nodes $v$, call them $v_{0},\ldots,v_{2^{\tau}-1}$.
To construct their lists $S_{v_i}$, we scan through $S_u$ and append $S_u[j]$
to the list of $v_{t}$, where $t$ is a bit string of length $\tau$ occurring in
the binary representation of $S_u[j]$ at position $\alpha\tau+1$. 
We can extract $t$ in constant time, so the total complexity is linear in the total
number of elements on all lists, i.e., $\Oh(n \log \sigma / \tau)$ if $\tau\le \log \sigma$.
Otherwise the root is the only big node, and the first phase is void.

\paragraph{Second phase.} Assume that we are given the list $S_u$ for a big node
 $u$ at distance $\alpha \tau$ from the root. We would like to
construct the bitmask $B_{v}$ for every node $v$ such that $u$ is the nearest big
ancestor of $v$. First, we observe that to construct all these bitmasks we only need
to know $\tau$ bits of every $S_u[j]$ starting from the $(\alpha \tau+1)$-th one. Therefore,
we will operate on \emph{short lists} $L_v$ consisting of $\tau$-bit integers instead of
the lists $S_{v}$ of whole $\log\sigma$-bit integers. Each short list is stored as a packed list.

We start with extracting the appropriate part of each $S_u[j]$ and appending  it to $L_u$.
The running time of this step is proportional to the length of $L_u$. (This step is void if $\tau>\log \sigma$; then $L_v=S_v$.)
Next, we process all nodes $v$ such that $u$ is the deepest big ancestor of $v$.
For every such $v$ we want to construct the short list $L_v$. Assuming that we already have $L_v$ for a node $v$
at distance $\alpha \tau + \beta$ from the root, where $\beta\in [0,\tau)$,
and we want to construct
the short lists of its children and the bitmask $B_{v}$. 
This can be (inefficiently) done by scanning $L_v$ and appending the next element to the short list of the right or the left child of $v$, depending on whether its $(\beta+1)$-th bit is set or not, respectively. The bitmask $B_{v}$ simply stores
all these $(\beta+1)$-th most significant bits. In order to compute
these three lists efficiently we apply the following claim.

\begin{claim}\label{clm:pct}
Assuming $\tilde\Oh(\sqrt{n})$ space and preprocessing shared by all instances of the structure, the following operation
can be performed in $\Oh(\frac{Nb}{\log n})$ time:
given a packed list $L$ of $N$ $b$-bit integers, where $b=o(\log n)$, and a position $t\in[0,b-1]$, compute
packed lists $L_0$ and $L_1$ consisting of the elements of $L$ whose $t$-th most significant bit is $0$ or $1$, respectively,
and a bitmask $B$ being a concatenation of the $t$-th most significant bits of the elements of $L$.
\end{claim}
\begin{proof}
As a preprocessing, we precompute all the answers for lists $L$ of length at most $\frac{1}{2}\frac{\log n}{b}$.
This takes $\tilde{\Oh}(\sqrt{n})$ time.
For a query we split $L$ into lists of length $\frac{1}{2}\frac{\log n}{b}$, apply the preprocessed mapping
and merge the results, separately for $L_0,L_1$ and $B$. This takes $\Oh(\frac{Nb}{\log n})$ time in total.
\mayqed \end{proof}

Consequently, we spend $\Oh(|L_v|\tau/\log n)$ per node $v$.
The total number of the elements of all short lists is $\Oh(n \log \sigma)$,
but we also need to take into the account the fact that the lengths of some short lists might
be not divisible by $\log n / \tau$, which adds $\Oh(1)$ time per a node of the
wavelet tree, making the total complexity $\Oh(\sigma + n \log \sigma \tau / \log n)=\Oh(n \log \sigma \tau / \log n)$.

\paragraph{Intermixing the phases.} 
In order to make sure that space usage of the construction algorithm does not exceed the size of the the final structure,
the phases are intermixed in the following way. Assuming that we are given the lists $S_u$
of all big nodes $u$ at distance $\alpha \tau$ from the root, we construct the lists
of all big nodes at distance $(\alpha+1)\tau$ from the root, if any. Then we construct
the bitmasks $B_{u}$ such that the nearest big ancestor of $u$ is at distance
$\alpha\tau$ from the root and increase $\alpha$. To construct the bitmasks,
we compute the short lists for all nodes at distances $\alpha\tau,\ldots,(\alpha+1)\tau-1$ from the
root and keep the short lists only for the nodes at the current distance.
Then the peak space of the construction process is just
$\Oh(n \log \sigma / \log n)$ words.

The total construction time is $\Oh(n \log \sigma / \tau)$ for the first phase and $\Oh(n \log \sigma \tau / \log n)$ for the second phase. 
The bitmasks, lists and short lists constructed by the algorithm are illustrated in Figure~\ref{fig:wavelet example2}.
Choosing $\tau=\sqrt{\log n}$ as to minimize the total time, we get the following theorem.

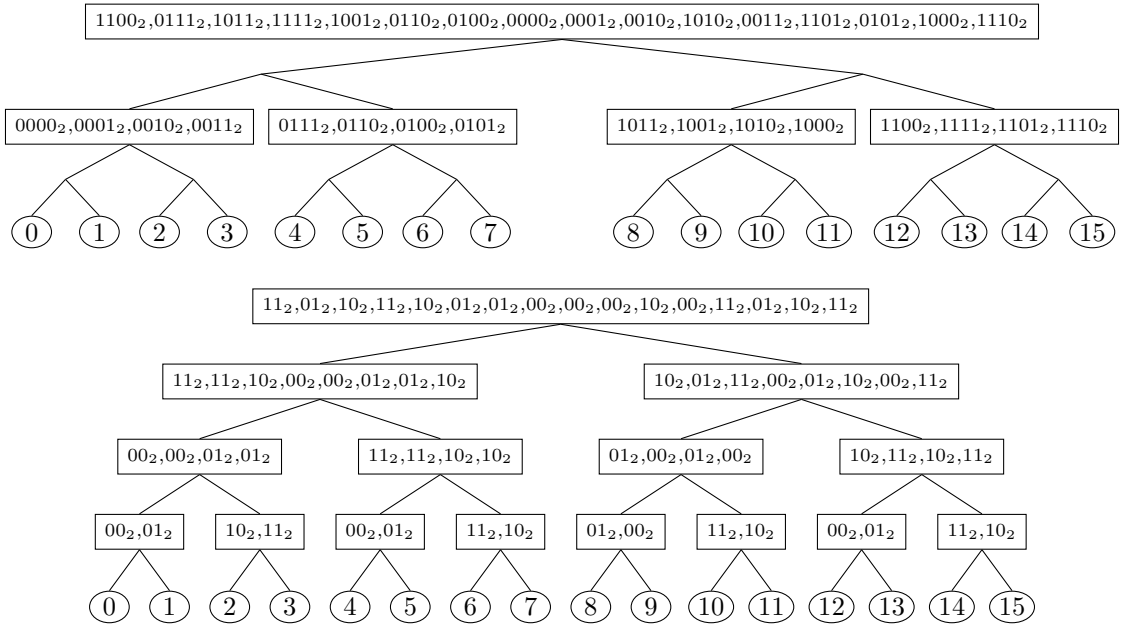
\begin{figure}
\begin{center}
\begin{tikzpicture}[level/.style={level distance=7mm,sibling distance=80mm/#1},level 2/.style={sibling distance=35mm},level 3/.style={sibling distance=17mm}, level 4/.style={sibling distance=9mm},minimum width=15pt, every node/.style={font=\small}]
\node [draw] {$\scriptstyle 1100_2,0111_2,1011_2,1111_2,1001_2,0110_2,0100_2,0000_2,0001_2,0010_2,1010_2,0011_2,1101_2,0101_2,1000_2,1110_2$}
child {
  {}
 child {
  node [draw] {$\scriptstyle 0000_2,0001_2,0010_2,0011_2$}
  child {
    {}
   child {
    node [ellipse,draw,inner sep=1pt] {$0$}
   }
   child {
    node [ellipse,draw,inner sep=1pt] {$1$}
   }
  }
  child {
    {}
   child {
    node [ellipse,draw,inner sep=1pt] {$2$}
   }
   child {
    node [ellipse,draw,inner sep=1pt] {$3$}
   }
  }
 }
 child {
  node [draw] {$\scriptstyle 0111_2,0110_2,0100_2,0101_2$}
  child {
    {}
   child {
    node [ellipse,draw,inner sep=1pt] {$4$}
   }
   child {
    node [ellipse,draw,inner sep=1pt] {$5$}
   }
  }
  child {
    {}
   child {
    node [ellipse,draw,inner sep=1pt] {$6$}
   }
   child {
    node [ellipse,draw,inner sep=1pt] {$7$}
   }
  }
 }
}
child {
  {}
 child {
  node [draw] {$\scriptstyle 1011_2,1001_2,1010_2,1000_2$}
  child {
    {}
   child {
    node [ellipse,draw,inner sep=1pt] {$8$}
   }
   child {
    node [ellipse,draw,inner sep=1pt] {$9$}
   }
  }
  child {
    {}
   child {
    node [ellipse,draw,inner sep=1pt] {$10$}
   }
   child {
    node [ellipse,draw,inner sep=1pt] {$11$}
   }
  }
 }
 child {
  node [draw] {$\scriptstyle 1100_2,1111_2,1101_2,1110_2$}
  child {
    {}
   child {
    node [ellipse,draw,inner sep=1pt] {$12$}
   }
   child {
    node [ellipse,draw,inner sep=1pt] {$13$}
   }
  }
  child {
    {}
   child {
    node [ellipse,draw,inner sep=1pt] {$14$}
   }
   child {
    node [ellipse,draw,inner sep=1pt] {$15$}
   }
  }
 }
};
\end{tikzpicture}

\vspace{.5cm}

\begin{tikzpicture}[level/.style={level distance=10mm,sibling distance=64mm/#1},level 2/.style={sibling distance=32mm},level 3/.style={sibling distance=16mm}, level 4/.style={sibling distance=8mm},minimum width=15pt, every node/.style={font=\small}]
\node [draw] {$\scriptstyle 11_2,01_2,10_2,11_2,10_2,01_2,01_2,00_2,00_2,00_2,10_2,00_2,11_2,01_2,10_2,11_2$}
child {
 node [draw] {$\scriptstyle 11_2,11_2,10_2,00_2,00_2,01_2,01_2,10_2$}
 child {
  node [draw] {$\scriptstyle 00_2,00_2,01_2,01_2$}
  child {
   node [draw] {$\scriptstyle 00_2,01_2$}
   child {
    node [ellipse,draw,inner sep=1pt] {$0$}
   }
   child {
    node [ellipse,draw,inner sep=1pt] {$1$}
   }
  }
  child {
   node [draw] {$\scriptstyle 10_2,11_2$}
   child {
    node [ellipse,draw,inner sep=1pt] {$2$}
   }
   child {
    node [ellipse,draw,inner sep=1pt] {$3$}
   }
  }
 }
 child {
  node [draw] {$\scriptstyle 11_2,11_2,10_2,10_2$}
  child {
   node [draw] {$\scriptstyle 00_2,01_2$}
   child {
    node [ellipse,draw,inner sep=1pt] {$4$}
   }
   child {
    node [ellipse,draw,inner sep=1pt] {$5$}
   }
  }
  child {
   node [draw] {$\scriptstyle 11_2,10_2$}
   child {
    node [ellipse,draw,inner sep=1pt] {$6$}
   }
   child {
    node [ellipse,draw,inner sep=1pt] {$7$}
   }
  }
 }
}
child {
 node [draw] {$\scriptstyle 10_2,01_2,11_2,00_2,01_2,10_2,00_2,11_2$}
 child {
  node [draw] {$\scriptstyle 01_2,00_2,01_2,00_2$}
  child {
   node [draw] {$\scriptstyle 01_2,00_2$}
   child {
    node [ellipse,draw,inner sep=1pt] {$8$}
   }
   child {
    node [ellipse,draw,inner sep=1pt] {$9$}
   }
  }
  child {
   node [draw] {$\scriptstyle 11_2,10_2$}
   child {
    node [ellipse,draw,inner sep=1pt] {$10$}
   }
   child {
    node [ellipse,draw,inner sep=1pt] {$11$}
   }
  }
 }
 child {
  node [draw] {$\scriptstyle 10_2,11_2,10_2,11_2$}
  child {
   node [draw] {$\scriptstyle 00_2,01_2$}
   child {
    node [ellipse,draw,inner sep=1pt] {$12$}
   }
   child {
    node [ellipse,draw,inner sep=1pt] {$13$}
   }
  }
  child {
   node [draw] {$\scriptstyle 11_2,10_2$}
   child {
    node [ellipse,draw,inner sep=1pt] {$14$}
   }
   child {
    node [ellipse,draw,inner sep=1pt] {$15$}
   }
  }
 }
};
\end{tikzpicture}
\vspace{-.7cm}
\end{center}
\caption{Elements of the wavelet tree construction algorithm for the string from Figure~\ref{fig:wavelet example1}. The first figure shows the lists $S_{u}$ of all big nodes $u$ when $\tau=2$, and the third shows the short lists $L_{u}$ of all nodes $u$, as defined in the proof of Theorem~\ref{th:wavelet_tree}.}
\label{fig:wavelet example2}
\end{figure}

\begin{theorem}\label{th:wavelet_tree}
Given a string $s$ of length $n$ over an alphabet $\Sigma$, we can construct all
bitmasks $B_{u}$ of its wavelet tree in $\Oh(n\log \sigma / \sqrt{\log n})$ time.
\end{theorem}

Additionally, we want to build a rank/select structure for every $B_{u}$. While it
is well-known that given a bit string of length $N$ one can construct
an additional structure of size $o(N)$ allowing executing both rank
and select in $\Oh(1)$ time~\cite{Clark,Jacobson}, we must verify that the construction
time is not too high. The following lemma is proved in Appendix~\ref{app:wt}.

\begin{lemma}\label{lem:rank_select}
Given a bit string $B[1..N]$ packed in $\frac{N}{\log n}$ machine
words, we can extend it in $\Oh(\frac{N}{\log n})$ time with a rank/select
structure occupying additional $o(\frac{N}{\log n})$ space, assuming
an $\tilde\Oh(\sqrt{n})$ time and space preprocessing shared by all instances
of the structure.
\end{lemma}

\subsection{Arbitrarily-Shaped Wavelet Trees}\label{ssec:aswt}
To generalize the algorithm to arbitrarily-shaped wavelet trees of degree $\Oh(\log\sigma)$,
instead of operating on the characters we work with the labels of the root-to-leaf paths, appended with
{\bf 0}s so that they all have the same length. The heap-like numbering of the nodes 
is not enough in such setting, as we cannot guarantee that all leaves have the same depth, so
the first step is to show how to efficiently return a node given the length of its label
and the label itself stored in a single machine word.
If $\log \sigma = o(\log n)$ we can afford storing nodes in a simple array, otherwise we use the
deterministic dictionary of Ru\v{z}i\'{c}~\cite{DBLP:conf/icalp/Ruzic08}, which can be constructed
in $\Oh(\sigma(\log\log \sigma)^{2})=\Oh(n(\log\log n)^{2})$ time. In either case, we can
return in $\Oh(1)$ time the corresponding pointer, which might be null if the node
does not exist. Then the construction algorithm works as described previously:
first we construct the list $S_{u}$ of every big node $u$, and then we construct
every $B_{v}$ using the $S_{u}$ of the nearest big ancestor of $v$. The only difference
is that when splitting the list $S_{u}$ into the lists $S_{v}$ of all big nodes $v$ such that
$u$ is the first proper big ancestor of $v$, it might happen that the retrieved pointer
to $v$ is null. In such case we simply continue without appending anything to the
non-existing $S_{v}$. The running time stays the same.

\begin{theorem}\label{thm:aswt}
Let $s$ be a string of length $n$ over~$\Sigma$ and
$T$ be a full binary tree of height $\Oh(\log \sigma)$ with $\sigma$ leaves, each assigned a distinct character in $\Sigma$.
Then the $T$-shaped wavelet tree of $s$ can be constructed in $\Oh(n\log\sigma/\sqrt{\log n})$ time.
\end{theorem}

\subsection{Wavelet Trees of Larger Degree}
\label{sec:higher degree construction}

We move to constructing a wavelet tree of degree $d=\log^{\eps}n$, where $d$ is a power of two. Such higher degree tree can
be also defined through the binary wavelet tree for $s$ as follows.
We remove all inner nodes whose depth is not a multiple of $\log d$. For each surviving node we set its nearest preserved ancestor as a parent. Then each inner node
has $d$ children (the lowest-level nodes may have fewer children), and we order them consistently with the left-to-right order in the original tree.

Recall that for each node $u$ of the binary wavelet tree we define the string $S_u$ as a subsequence of $s$
consisting of its characters whose binary representation starts with the label $\ell_u$ of $u$. Then we create the bitmask
storing, for each character of $S_u$, the bit following its label $\ell_u$. Instead of $B_u$ a node $u$ of the wavelet tree of
degree $d$ now stores a string $D_u$, which contains the next $\log d$ bits following $\ell_u$.

The following lemma allows to use binary wavelet tree construction as a black box, and consequently gives an $\Oh(n\sqrt{\log n})$-time construction
algorithm for wavelet trees of degree~$d$.

\begin{lemma}
\label{lm:rebuilding}
Given all bitmasks $B_{u}$, we can construct all strings $D_{u}$ in $\Oh(n\log\log n)$
time.
\end{lemma}

\begin{proof}
Consider a node $u$ at depth $k$ of the wavelet tree of degree $d$. Its children
correspond to all descendants of $u$ in the original wavelet tree at depth $(k+1)d$.
For each descendant $v$ of $u$ at depth $(k+1)\log d-\delta$, where $\delta\in [0,\log d]$,
we construct a temporary string $D'_{v}$ over the alphabet $[0,2^{\delta}-1]$. Every character
of this temporary string corresponds to a leaf in the subtree of $v$. The characters are arranged
in order of the identifiers of the corresponding leaves and describe prefixes of length $\delta$
of paths from $u$ to the leaves. Clearly $D_{u}=D'_{u}$. Furthermore, if the children of $v$ are
$v_{1}$ and $v_{2}$, then $D'_{v}$ can be easily defined by looking at $D'_{v_{1}}$, $D'_{v_{2}}$,
and $B_{v}$ as follows. We prepend {\bf 0} to all characters in $D'_{v_{1}}$, and {\bf 1} to all characters in $D'_{v_{2}}$.
Then we construct $D'_{v}$ by appropriately interleaving $D'_{v_{1}}$ and $D'_{v_{2}}$ according
to the order defined by $B_{v}$. We consider the bits of $B_{v}$ one-by-one. If the $i$-th bit is
{\bf 0}, we append the next character of $D'_{v_{1}}$ to the current $D'_{v}$, and otherwise we append
the next character of $D'_{v_{2}}$. Now we only have to show to implement this process efficiently.

We pack every $\frac{1}{4}\frac{\log n}{\log d}$ consecutive characters of $D'_{v}$ into a single machine word.
To simplify the implementation, we reserve $\log d$ bits for every character irrespective of the value of $\delta$.
This makes prepending {\bf 0}s or {\bf 1}s to all characters in any $D'_{v}$ trivial, because the result
can be preprocessed in $\tilde\Oh(d^{\frac{1}{4}\frac{\log n}{\log d}})=o(n)$ time and space. 
Interleaving $D'_{v_{1}}$ and $D'_{v_{2}}$ is more complicated. Instead of accessing $D'_{v_{1}}$ and $D'_{v_{2}}$ directly,
we keep two buffers, each containing at most next $\frac{1}{4}\frac{\log n}{\log d}$ characters from the
corresponding string. Similarly, instead of accessing $B_{v}$ directly, we keep
a buffer of at most $\frac{1}{4}\log n$ next bits there.
Using the buffers, we can keep processing bits from $B_{v}$ as long as there are enough characters
in the input buffers. The input buffers for $D'_{v_{1}}$ and $D'_{v_{2}}$ become
empty after generating $\frac{1}{4}\frac{\log n}{\log d}$ characters, and the input buffer for $B_{v}$ becomes empty after
generating $\frac{1}{4}\log n$ characters. Hence the total number of times we need to reload one of the 
input buffers is $\Oh(|D'_{v}|/ \frac{\log n}{\log d})$.

We preprocess all possible scenarios between two reloads by simulating, for every possible initial
content of the input buffers, processing the bits until one
of the buffers becomes empty.
We store the generated data (which is at most $\frac{1}{2}\log n$ bits) and the final
content of the input buffers. The whole preprocessing takes $\tilde\Oh(2^{\frac{3}{4}\log n})=o(n)$ time and space,
and then the number of operations required to generate packed
$D'_{v}$ is proportional to the number of times we need to reload the buffers, so by summing over
all~$v$ the total complexity is $\Oh(n\log\log n)$.
\mayqed \end{proof}

Then we extend every $D_{u}$ with a generalized rank/select data structure.
Such a structure for a string $D[1..N]$ implements operations $\rank_{c}(i)$, which counts
positions $k \in [1,i]$ such that $D[k] \leq c$, and $\select_{c}(i)$, which selects the
$i$-th occurrence of $c$ in the whole $D[1..n]$, for any $c\in\Sigma$, both in $\Oh(1)$ time.
Again, it is well-known that such a structure can be implemented using just $o(n\log \sigma)$
additional bits if $\sigma=\Oh(\polylog(n))$~\cite{Ferragina}, but its construction time
is not explicitly stated in the literature. Therefore, we prove the following lemma in Appendix~\ref{app:wt}.

\begin{lemma}
\label{lm:rank/select construction 2}
Let $d\le \log^{\eps} n$ for $\eps<\frac13$.
Given a string $D[1..N]$ over the alphabet $[0, d-1]$ packed in $\frac{N\log d}{\log n}$ machine
words, we can extend it in $\Oh(\frac{N\log d}{\log n})$ time with a generalized rank/select data
structure occupying additional $o(\frac{N}{\log n})$ space, assuming $\tilde\Oh(\sqrt{n})$ time and space preprocessing shared by all instances
of the structure.
\end{lemma}

\subsection{Range Selection}
\label{sec:range selection}
A classic application of wavelet trees is that, given an array $A[1..n]$ of integers,
we can construct a structure of
size $\Oh(n)$, which allows answering any range rank/select query in $\Oh(\log n)$
time. A range select query is to return the $k$-th smallest
element in $A[i..j]$, while a range rank query is to count
how many of $A[i..j]$ are smaller than given $x = A[k]$.
Given the wavelet tree for $A$, 
any range rank/select query can be answered by traversing a root-to-leaf path of the tree using the rank/select data structures for bitmasks $B_u$ at subsequent nodes.

With $\Oh(n\sqrt{\log n})$ construction algorithm this matches the bounds of Chan and
P\u{a}tra\c{s}cu~\cite{DBLP:conf/soda/ChanP10} for range select queries, but is slower by a factor of $\log\log n$
than their solution for range rank queries.
We will show that one can in fact answer
any range rank/select query in $\Oh(\frac{\log n}{\log\log n})$ time with an $\Oh(n)$
size structure, which can be constructed in $\Oh(n\sqrt{\log n})$ time.
For range rank queries this is not a new result, but we feel that our proof gives more
insight into the structure of the problem. For range select queries, Brodal et al.~\cite{BrodalMedian}
showed that one can answer a query in $\Oh(\frac{\log n}{\log\log n})$ time
with an $\Oh(n)$ size structure, but with $\Oh(n\log n)$ construction time.
Chan and P\u{a}tra\c{s}cu asked if methods of \cite{BWT} can be combined with the
efficient construction. We answer this question affirmatively. 

A range rank query is easily implemented in $\Oh(\frac{\log n}{\log\log n})$ time using
wavelet tree of degree $\log^{\eps}n$ described in the previous section.
To compute the rank of $x = A[k]$ in $A[i..j]$, we traverse the path from the root to the leaf
corresponding to $A[k]$. At every node
we use the generalized rank structure to update the answer and the current interval $[i..j]$ before
we descend. 

Implementing the range select queries in $\Oh(\frac{\log n}{\log \log n})$ time is more complicated. Similar to the rank queries, we start the traverse at the root and descend along a root-to-leaf path. At each node we select its child we will descend to, and update the current interval $[i..j]$ using the generalized rank data structure. To make this query algorithm fast, we preprocess each string $D_u$ and store extracted information in matrix form. As shown by Brodal et al.~\cite{BrodalMedian}, constant time access to such information is enough to implement range select queries in $\Oh(\log n/\log\log n)$ time.

The matrices for a string $D_u$ are defined as follows. We partition $D_u$ into superblocks of length $d\log^2 n$.\footnote{In the original paper, superblocks are of length $d\log n$, but this does not change the query algorithm.} For each superblock we store the
cumulative generalized rank of every character, i.e., for every character $c$ we store the number positions
where characters $c'\le c$ occur in the prefix of the string up to the beginning of the superblock. We think of this as a $d\times \log n$
matrix $M$. 
The matrix is stored in two different ways. In the first copy, every row is stored as a single word. In the second copy, we divide the matrix into sections of $\log n / d$ columns, and store every section in a single word. We make sure there is an overlap of four columns between the sections, meaning that the first section contains columns $1,\ldots,\log n / d$, the second section contains columns $\log n / d -3,\ldots, 2\log n / d-4$, and so on. 

Each superblock is then partitioned into blocks of length $\log n/\log d$ each. 
For every block, we store the cumulative generalized rank within the superblock
of every character, i.e., for every
character $c$ we store the number of positions where characters $c'\le c$ occur in the prefix
of the superblock up to the end beginning of the block. 
We represent this information in a \emph{small matrix}~$M'$,
which can be packed in a single word, as it requires only $\Oh(d \log (d\log^2 n))$ bits.

\begin{lemma}
\label{lm:Brodal construction}
Given a string $D[1..N]$ over the alphabet $[0, d-1]$ packed in $\frac{N\log d}{\log n}$
machine words, we can extend it in $\Oh(\frac{N\log d}{\log n})$ time and space with the following information:
\begin{enumerate}[1)]\compact
\item The first copy of the matrix $M$ for each superblock (multiple of $d\log^2 n$);
\item The second copy of the matrix $M$ for each superblock (multiple of $d\log^2 n$);
\item The small matrix $M'$ for each block (multiple of $\log n/d$);
\end{enumerate}
occupying additional $o(\frac{N\log d}{\log n})$ space, assuming
an $\tilde\Oh(\sqrt{n})$ time and space preprocessing shared by all instances
of the structure and $d =\log^{\eps}n$.
\end{lemma}
\begin{proof}
To compute the small matrices $M'$, i.e., the cumulative generalized ranks for blocks,
and the first copies of the matrices $M$, i.e., the cumulative generalized ranks for superblocks,
we notice that the standard solution for generalized rank queries in $\Oh(1)$ time is to split
the string into superblocks and blocks. Then, for every superblock we store the cumulative generalized
rank of every character, and for every block we store the cumulative generalized rank within
the superblock for every character. As explained in the proof of Lemma~\ref{lm:rank/select construction 2}
presented in the appendix, such data can be computed in $\Oh(\frac{N\log d}{\log n})$ time
if the size of the blocks and the superblocks are chosen to be $\frac{1}{3}\frac{\log n}{\log d}$
and $d\log^2n$, respectively. Therefore, we obtain in the same complexity every
first copy of the matrix $M$, and a small matrix every $\frac{1}{3}\frac{\log n}{\log d}$ characters.
We can simply extract and save every third such small matrix, also in the same complexity.

The second copies of the matrices are constructed from the first copies in $\Oh(d^{2})$ time each;
 we simply partition each row into (overlapping) sections and append each part to the appropriate machine words.
 This takes $\Oh(\frac{N d^{2}}{d^{2}\log n})=\Oh(\frac{N}{\log n})$ in total.
 \mayqed \end{proof}

As follows from the lemma, all strings $D_u$ at a single level of the tree can be extended in $\Oh(n \log \log n / \log n)$ time,
which, together with constructing the tree itself, sums up to $\Oh(n\sqrt{\log n})$ time in total. 

\subsection{Range Successor Queries}
\label{sec:orpq}
In this section we show how wavelet trees can be used to answer range successor queries.
In our setting, these queries can be interpreted as follows: Given a range $R=[i,j]$ of positions of a string $s$ and a character $c$ compute $c'=\min\{s[k] : i\in R, s[k]\ge c\}$, the successor of $c$ in $R$.

\subsubsection{Online algorithm}
We first show how to answer successor queries online by a straightforward reduction to range rank and select queries.

\begin{theorem}\label{thm:rsq:on}
Given an array $A[1..n]$ of integers, in $\Oh(n\sqrt{\log n})$ time we can construct a data structure of size $\Oh(n)$,
which can answer range successor queries in $\Oh(\log n / \log \log n)$ time.
\end{theorem}
\begin{proof}
To compute the successor of $A[k]$ in $A[i..j]$, we determine the rank $r$ of $A[k]$ in $A[i..j]$ and use a range selection query
to find the $(r+1)$-th smallest element in $A[i..j]$. This element is the successor of $A[k]$.
\end{proof}

\subsubsection{Offline Algorithm}
We begin with a simple (though inefficient) online algorithm which uses just the wavelet tree.
Later we will show how to speed it up by making it offline.
Recall that for each node $v$ of the wavelet tree we defined $S_v$ as the subsequence of $s$ containing the characters corresponding to the leaves in the subtree rooted at $v$.
For a fixed range $R$ let $R_v$ be the \emph{inherited} range of positions, i.e., a range of positions in $S_v$ such that the corresponding positions of $s$ belong to $R$.
Note that if $v$ is the root of the wavelet tree, then $R_v = R$.
Moreover, if $w$ is the parent of $v$, we can use $R_w$ and the rank structure on the bitmask $B_w$ to compute $R_v$ in constant time.
The problem of efficiently determining $R_v$ for a given $R$ and $v$, known as \emph{ball inheritance problem}~\cite{DBLP:journals/siamcomp/Chazelle88,DBLP:conf/compgeom/ChanLP11}, is the heart of state-of-the-art online algorithms for range successor queries~\cite{SortedRangeReporting,DBLP:journals/corr/Zhou13b}.
Our solutions do not rely this tool as a black box but computing and maintaining inherited ranges 
is still their important ingredient.

We identify leaves of the wavelet trees with characters in $\Sigma$. 
Note that $c'$ is the leftmost leaf to the right of $c$ with a non-empty inherited range $R_{c'}$.
While computing $c'$, as an intermediate step we determine the lowest common ancestor $w$ of $c$ and $c'$.
If $c'\ne c$,  then $w$ is the deepest ancestor of $c$ whose left subtree contains $c$ and right child
$w_{r}$ satisfies $R_{w_r}\ne \emptyset$.
Moreover, $c'$ is then the leftmost leaf in the subtree of $w_r$ for which the inherited range is non-empty.

This leads to the following algorithm: we traverse the path from the root to $c$, maintaining the inherited ranges.
Whenever we visit a node $w$ such that $c$ is in its left subtree, we check whether
$R_{w_r}=\emptyset$ where $w_r$ is the right child of $w$. During the algorithm we remember the last node $w$ satisfying this property and the inherited range $R_w$.
If $R_c\ne \emptyset$, then $c'=c$ is the successor of $c$ in $R$. Otherwise, we find the leftmost leaf in the subtree of $w_r$ with a non-empty inherited range.
For this we descend the tree from $w_r$ going to the left child if its inherited range is non-empty and to the right child otherwise.
In both phases of the algorithm we follow a single path down the tree, so the running time is $\Oh(\log \sigma)$.

To speed up the algorithm, we reuse the concept of big nodes, introduced in Section~\ref{sec:WTconstruction}.
Recall that a node is big if its depth is a multiple of $\tau = \sqrt{\log n}$ and small otherwise.
The construction algorithm explicitly generates subsequences $S_v$ for all big nodes $v$.
To answer range successor queries, we store these subsequences and augment them with data structures
for range minimum and range maximum queries~\cite{DBLP:journals/siamcomp/HarelT84,Bender:2000:LPR:646388.690192}.

Improving the running time of the second part of the algorithm is easy: 
once we visit a big node $v$, instead of proceeding down to the leaf, we ask a range minimum query for the range $R_v$ of $S_v$
to determine the final answer. Consequently, this part takes $\Oh(\tau)$ time only.

On the other hand, range maximum queries let us easily test if a given big ancestor $v$ of $c$ is also an ancestor of $c'$:
we check whether the maximum in the range $R_v$ of $S_v$ is at least $c$. Thus, if we knew inherited ranges for all
$\Oh(\frac{\log \sigma}{\tau})$ big ancestors of $c$, we could determine the lowest of them which is an ancestor of $c'$
and run the first part of the original algorithm starting at that ancestor and terminating at the next big ancestor (or at $c'$, if none).
Hence, we could determine the lowest common ancestor $w$ in $\Oh(\frac{\log \sigma}{\tau}+\tau)$ time.

Computing the inherited ranges of big ancestors is where we benefit from the offline setting.
\begin{lemma}
Given a collection of $q$ queries consisting of a range  $R_i$ and a character $c_i$,
we can compute the inherited ranges of all big ancestors for all $c_i$ in $\Oh((n+q)\frac{\log \sigma}{\tau})$ time.
\end{lemma}
\begin{proof}
We only show how to find the right endpoints of the ranges. Computing the left endpoints is analogous.
The algorithm resembles the first phase of the wavelet tree construction.

We compute the endpoints in a level-by-level top-down manner.
Let $R_{i,u}$ be the range induced by $R_i$ at a big node $u$.
In a single step we use the endpoints of the ranges $R_{i,u}$ at the node $u$ to find the endpoints of the ranges at nodes whose deepest (proper) big ancestor is $u$. 
There are exactly $2^\tau$ such nodes $v_0,\ldots,v_{2^\tau-1}$.
We process the sequence $S_u$ maintaining a counter for each node $v_t$, $t = 0..2^\tau-1$.
For each $j$ we first increment the counter of the node $v_t$ such that $S_{v_t}$ contains $S_u[j]$.
Then, we process the ranges $R_{i,u}$ whose right endpoint is $j$. We identify the node $v_t$ that is the ancestor of $c_i$ and set the endpoint of $R_{i,v_t}$ to be the current value of the counter for $v_t$.

Such a single step takes $\Oh(|S_u|+q_u)$ time where $q_u$ is the number of queries with $c_i$ in the subtree of $u$.
This sums up to $\Oh(n+q)$ per level and $\Oh((n+q)\frac{\log \sigma}{\tau})$ in total.
\mayqed
\end{proof}
The total running time, including wavelet tree construction, is $\Oh(n\frac{\tau \log \sigma}{\log n} + (n+q)\frac{\log \sigma}{\tau} + q\tau)$.
With $\tau = \sqrt{\log \sigma}$, this gives $\Oh((n+q)\sqrt{\log \sigma} + n\log^{3/2} \sigma/\log n)$ time and yields the following result.

\begin{theorem}
A collection of $q$ range successor queries on a string $s$ of length $n$ over an alphabet~$\Sigma$
can be answered in $\Oh((n+q)\sqrt{\log \sigma})$ time.
\end{theorem}

\begin{corollary}\label{cor:rsq:off}
Given an array $A[1..n]$ of $n$ integers, we can answer $q$ range successor queries in $\Oh((n+q)\sqrt{\log n})$ total time.
\end{corollary}
\begin{proof}
A standard reduction lets us assume that values in the array are in $[0,n-1]$.
This is at the price of sorting the array, which takes $\Oh((n+q)\log \log n)$ time if one uses a deterministic sorting algorithm by Han~\cite{Han}.
\end{proof}

\section{Wavelet Suffix Trees}\label{sec:wst}
In this section we generalize wavelet trees to obtain wavelet suffix trees.
With logarithmic height and shape resembling the shape of the suffix tree wavelet suffix trees, augmented with additional stringological data structures, become a very powerful tool. In particular, they allow to answer the following queries efficiently: (1) find the $k$-th lexicographically minimal suffix of a substring of the given string (\emph{substring suffix selection}), (2) find the rank of one substring among the suffixes of another substring (\emph{substring suffix rank}), and (3) compute the run-length encoding of the Burrows-Wheeler transform of a substring. 

\paragraph{Organisation of Section~\ref{sec:wst}.}
In Section~\ref{ssec:prelim} we introduce several, mostly standard, stringological notions
and recall some already known algorithmic results.
Section~\ref{ssec:over} provides a high-level description of the wavelet suffix trees. It forms an interface
between the query algorithms (Section~\ref{ssec:apps}) and the more technical
content: full description of the data structure (Section~\ref{ssec:full}) and its construction algorithm (Section~\ref{ssec:constr}).
Consequently, Sections~\ref{ssec:full} and~\ref{ssec:apps} can be read separately.
The latter additionally contains cell-probe lower bounds for some of the queries (suffix rank \& selection),
as well as a description of a generic transformation of the data structure, which allows to replace a dependence on $n$ with a dependence on $|x|$
in the running times of the query algorithms.

\subsection{Preliminaries}\label{ssec:prelim}
Let $w$ be a string of length $|w| = n$ over the alphabet $\Sigma = [0, \sigma-1]$. 
For $1 \le i \le j \le n$, $w[i..j]$ denotes the \emph{substring} of $w$ from position $i$ to position $j$ (inclusive). 
For $i = 1$ or $j = |w|$, we use shorthands $w[..j]$ and $w[i..]$. 
If $x=w[i..j]$, we say that $x$ \emph{occurs} in $w$ at position $i$.
Each substring $w[..j]$ is called a \emph{prefix} of $w$, and each substring $w[i..]$ is called a \emph{suffix} of $w$. 
A substring which occurs both as a prefix and as a suffix of $w$ is called a \emph{border} of $w$.
The length longest common prefix of two strings $x,y$ is denoted by $\lcp(x,y)$.

We extend $\Sigma$ with a sentinel symbol $ \$$, which we assume to be smaller than any other character.
The order on $\Sigma$ can be generalized in a standard way to the \emph{lexicographic} order of the strings over $\Sigma$: a string $x$ is lexicographically smaller than $y$ (denoted $x\prec y$) if either $x$ is a proper prefix, or there exists a position~$i$, $0 \le i < \min\{|x|, |y|\}$, such that $x[1..i] = y[1..i]$ and $x[i+1] \prec y[i+1]$. 
The following lemma provides one of the standard tools in stringology. 
\begin{lemma}[LCP Queries \cite{AlgorithmsOnStrings}]\label{lem:all}
A~string $w$ of length $n$ can be preprocessed in $\Oh(n)$ time so that the following queries can be answered in $\Oh(1)$ time: Given substrings $x$ and $y$ of $w$, compute $\lcp(x, y)$ and decide whether $x \prec y$, $x = y$, or $x \succ y$.
\end{lemma}

We say that a sequence $p_0< p_1 < \ldots < p_k$ of positions in a string $w$
is a  \emph{periodic progression} if $w[p_0 .. p_{1}-1]=\ldots=w[p_{k-1}.. p_k-1]$.
Periodic progressions $p,p'$ are called \emph{non-overlapping} if the maximum term in $p$ is smaller than the minimum term in $p'$ or vice versa,  the maximum term in $p'$ is smaller than the minimum term in $p$.
Note that any periodic progression is an arithmetic progression and consequently it can be represented by three integers: $p_0$, $p_1-p_0$, and $k$.
 Periodic progressions appear in our work because of the following result: 
\begin{theorem}[\cite{DBLP:journals/corr/KociumakaRRW13}]\label{th:occurrences}
Using a data structure of size $\Oh(n)$ with $\Oh(n)$-time randomized (Las Vegas) construction, the following queries can be answered in constant time: Given two substrings $x$ and $y$ such that $|x|=\Oh(|y|)$, report the positions of all occurrences of $y$ in~$x$,
represented as at most $\frac{|x|+1}{|y|+1}$ non-overlapping periodic progressions.
\end{theorem}

\subsection{Overview of wavelet suffix trees}\label{ssec:over}
For a string $w$ of length $n$, a \emph{wavelet suffix tree} of $w$ is a full binary tree of logarithmic height. 
Each of its $n$ leaves corresponds to a non-empty suffix of $w\$$.
The lexicographic order of suffixes is preserved as the left-to-right order of leaves.

Each node $u$ of the wavelet suffix tree stores two bitmasks. 
Bits of the first bitmask correspond to suffixes below $u$ sorted by their starting positions,
and bits of the second bitmask correspond to these suffixes sorted by pairs (preceding character, starting position). 
The $i$-th bit of either bitmask is set to $0$ if the $i$-th suffix belongs to the left subtree of $u$ and to $1$ otherwise. 
Like in the standard wavelet trees, on top of the bitmasks we maintain a rank/select data structure.
See Figure~\ref{fig:wst} for a sample wavelet suffix tree with both bitmasks listed down in nodes.

Each edge $e$ of the wavelet suffix tree is associated with a sorted list $L(e)$ containing substrings of~$w$.
The wavelet suffix tree enjoys an important \emph{lexicographic property}.
Imagine we traverse the tree depth-first, and when going \emph{down} an edge $e$ we write out the contents of $L(e)$,
whereas when visiting a leaf we output the corresponding suffix of $w\$$.
Then, we obtain the lexicographically sorted list of all substrings of $w\$$ (without repetitions).%
\footnote{A similar property holds for suffix trees if we define $L(e)$ so that it contains the labels of
all implicit nodes on $e$ and the label of the lower explicit endpoint of $e$.}
This, in particular, implies that the substrings in $L(e)$ are consecutive prefixes of the longest substring in $L(e)$,
and that for each substring $y$ of $w$ there is exactly one edge $e$ such that the $y\in L(e)$.

\begin{figure}[t]
\begin{center}

\begin{tikzpicture}[every text node part/.style={align=left}]
\tikzset{every node/.style={font={\footnotesize \ttfamily},inner sep=0.5pt}, every internal node/.style={draw,rectangle, font=\scriptsize, inner sep = 2pt}, every leaf node/.style={draw,circle, font=\normalsize, minimum size=15pt}, level distance=1.4cm,sibling distance=.3cm}
\Tree[.\node (root) {0101101010110\\0111110100010};
   [.{111110\\ 111110} 13 \edge node[auto=left, pos=.4] {\color{red} a} node[auto=left, pos=.6]{ab} node[auto=left, pos=.8]{\color{red} aba};
	[.{01001\\ 01001}  \edge node[auto=right, pos=0.275] {abab} node[auto=right, pos=.5] {\color{red} ababa} node[auto=right, pos=.725] {ababab} node[auto=right, pos=.95] {\color{blue} abababb};
		[.{100\\ 100}
			  [.{01\\ 01} 6  \edge node[auto=left, pos=.6]{\color{blue} ababb}; 8 ]  \edge node[auto=left, pos=.3]{ababba} node[auto=left, pos=.525]{ababbab}node[auto=left, pos=.75]{ababbaba}  node[auto=left, pos=.9]{$\dots$}; 1
		]  \edge node[auto=left, pos=.6]{\color{blue} abb};
		[ .{10\\ 10} 10   \edge node[auto=left, pos=.3]{abba} node[auto=left, pos=.525]{abbab}node[auto=left, pos=.75]{abbaba}  node[auto=left, pos=.9]{$\dots$};  3 ]
	]
]  \edge node[auto=left, pos=.5]{\color{blue} b};
	[ .{0100010\\ 0100100}
		[ .{11110\\ 11110} 12   \edge node[auto=left, pos=.45]{\color{red} ba} node[auto=left, pos=.7]{bab};
			[ .{1001\\ 1010}  \edge  node[auto=right, pos=0.275] {\color{red} baba} node[auto=right, pos=.5] {babab} node[auto=right, pos=.725] {\color{red} bababa} node[auto=right, pos=.95] {bababab};
				[ .{01\\ 10} \edge node[auto=right, pos=.6]{babababb}; 5 \edge node[auto=left, pos=.6]{\color{blue} bababb}; 7 ]
				\edge node[auto=left,pos=.55]{\color{blue} babb}; [ .{10\\ 10}  9 \edge node[auto=left, pos=.3]{babba} node[auto=left, pos=.525]{babbab}node[auto=left, pos=.75]{babbaba}  node[auto=left, pos=.9]{$\dots$}; 2 ]
			]
		]
	 \edge node[auto=left]{\color{blue} bb}; [ .{10\\ 10} 11 \edge node[auto=left, pos=.3]{bba} node[auto=left, pos=.525] {bbab}node[auto=left, pos=.75]{bbaba}  node[auto=left, pos=.9]{$\dots$}; 4 ]
	]
];
\end{tikzpicture}
\vspace{-.5cm}
\end{center}
\caption{A wavelet suffix tree of $w=\texttt{ababbabababb}$. 
Leaves corresponding to $w[i..]\$$ are labelled with $i$.
Elements of $L(e)$ are listed next to $e$, with $\dots$
denoting further substrings up to the suffix of $w$. Suffixes of $x=\texttt{\color{red}bababa}$ are marked red, of $x=\texttt{\color{blue}abababb}$: blue. Note that the prefixes of $w[i..]$ do not need to lie above the leaf~$i$ (see $w[1,5]=\texttt{ababb}$),
and the substrings above the leaf $i$ do not need to be prefixes of $w[i..]$ (see $w[10..]$ and $\texttt{aba}$).}
\label{fig:wst}
\end{figure}
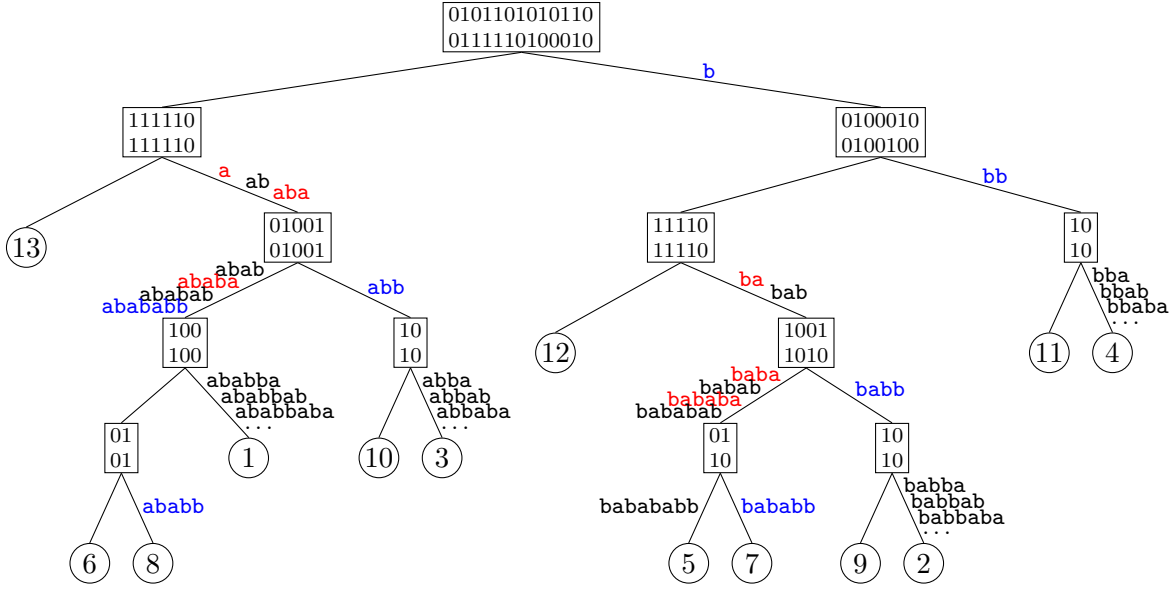

In the query algorithms, we actually work with $L_x(e)$, containing the suffixes of $x$ among the elements of $L(e)$.
For each edge $e$, starting positions of these suffixes form $\Oh(1)$ non-overlapping periodic progressions,
and consequently the list $L_x(e)$ admits a constant-space representation.
Nevertheless, we do not store the lists explicitly, but instead generate some of them on the fly.
This is one of the auxiliary operations,
each of which is supported by the wavelet suffix tree in constant~time.
\begin{enumerate}[(1)]\compact
  \item For a substring $x$ and an edge $e$, output the list $L_x (e)$ represented as $\Oh(1)$ non-overlapping periodic progressions;
  \item Count the number of suffixes of $x = w[i..j]$ in the left/right subtree of a node (given along with the segment of its first bitmask  corresponding to suffixes that start inside $[i, j]$);
  \item Count the number of suffixes $x=w[i..j]$ that are preceded by a character~$c$ and lie in the left/right subtree of a node (given along with the segment of its second bitmask corresponding to suffixes that start inside $[i, j]$ and are preceded by $c$);
  \item For a substring $x$ and an edge $e$, compute the run-length encoding of the sequence of characters preceding suffixes in $L_x(e)$.
\end{enumerate}

\subsection{Full description of wavelet suffix trees}\label{ssec:full}
We start the description with Section~\ref{ssec:tools}, where we introduce \emph{string intervals}, 
a notion central to the definition of wavelet suffix tree.  
We also present there corollaries of Lemma~\ref{lem:all} which let us efficiently deal with string intervals.
Then, in Section~\ref{ssec:WSTdefinition}, we give a precise definition of wavelet suffix trees and prove its several combinatorial consequences.
We conclude with Section~\ref{ssec:ops}, where we provide the implementations of auxiliary operations defined in Section~\ref{ssec:over}.

\subsubsection{String intervals}\label{ssec:tools}
To define wavelet suffix trees, we often need to compare substrings of $w$ trimmed to a certain number of characters. 
If instead of $x$ and $y$ we compare their counterparts trimmed to $\ell$ characters,
i.e., $x[1..\min\{\ell,|x|\}]$ and $y[1..\min\{\ell,|y|\}]$,
we use $\ell$ in the subscript of the operator, e.g., $x=_{\ell}y$  or $x \preceq_\ell y$.

For a pair of strings $s, t$ and a positive integer $\ell$ we define a \emph{string interval} $[s, t]_\ell=\{z \in \bar{\Sigma}^*: s \preceq_\ell z \preceq_\ell t\}$ and $(s,t)_\ell = \{z\in \bar{\Sigma}^* : s\prec_\ell z \prec_\ell t\}$. Intervals $[s, t)_\ell$ and $(s, t]_\ell$ are defined analogously. The strings $s,t$ are called  the \emph{endpoints} of these intervals.

In the remainder of this section, we show that the data structure of Lemma~\ref{lem:all} can answer queries
related to string intervals and periodic progressions, which arise in Section~\ref{ssec:ops}.
We start with a simple auxiliary result; here $y^\infty$ denotes a (one-sided) infinite string
obtained by concatenating an infinite number of copies of $y$.

\begin{lemma}\label{lem:corr}
The data structure of Lemma~\ref{lem:all} supports the following queries in $\Oh(1)$ time:
\begin{enumerate}[(1)]\compact
  \item\label{it:lexStrong} Given substrings $x$, $y$ of $w$ and an integer $\ell$, determine if $x \prec_\ell y$, $x =_\ell y$, or $x \succ_\ell y$.
  \item\label{it:pp} Given substrings $x,y$ of~$w$,
  compute $\lcp(x,y^\infty)$ and determine whether $x\prec y^\infty$ or $x\succ y^\infty$.
\end{enumerate}
\end{lemma}
\begin{proof}
\noindent
(\ref{it:lexStrong})  By Lemma~\ref{lem:all}, we may assume to know $\lcp(x,y)$. If  $\lcp(x,y) \ge \ell$, then $x =_\ell y$. Otherwise, 
trimming $x$ and $y$ to $\ell$ characters does not influence the order between these two substrings. 

\smallskip

\noindent
(\ref{it:pp}) 
If $\lcp(x,y)< |y|$, i.e., $y$ is not a prefix of $x$,
then $\lcp(x,y^\infty)=\lcp(x,y)$ and the order between $x$ and $y^\infty$ is the same as between $x$ and $y$.
Otherwise, define $x'$ so that $x=yx'$. Then $\lcp(x,y^\infty) = |y|+\lcp(x',x)$ and the order
between $x$ and $y^\infty$ is the same as between $x'$ and~$x$.
Consequently, the query can be answered in constant time in both cases.
\mayqed
\end{proof}

\begin{lemma}\label{lem:IntervalSelection}
The data structure of Lemma~\ref{lem:all} supports the following queries in $\Oh(1)$ time:
Given a periodic progression $p_0<\ldots<p_k$ in $w$, a position $j\ge p_k$, and a string interval $I$ whose endpoints are substrings of~$w$,
report, as a single periodic progression, all positions $p_i$ such that $w[p_i..j] \in I$.
\end{lemma}
\begin{proof}
If $k=0$, it suffices to apply Lemma~\ref{lem:corr}(\ref{it:lexStrong}).
Thus, we assume $k\ge 1$ in the remainder of the proof.

Let $s$ and $t$ be the endpoints of $I$, $\rho=w[p_0..p_1-1]$, and $x_i = w[p_i..j]$. Using Lemma~\ref{lem:corr}(\ref{it:pp})
we can compute $r_0=\lcp(x_0, \rho^\infty)$ and $r'=\lcp(s,\rho^\infty)$.
Note that $r_i := \lcp(x_i, \rho^\infty)=r-i|\rho|$, in particular $r_0\ge k|\rho|$.
If $r'\ge \ell$, we distinguish two cases:
\begin{enumerate}[1)]\compact
  \item $r_i \ge \ell$. Then $\lcp(x_i, s)\ge \ell$; thus $x_i =_\ell s$.
  \item $r_i < \ell$. Then $\lcp(x_i,s)=r_i$; thus $x_i\prec_\ell s$ if $x_0 \prec \rho^\infty$, and $x_i\succ_\ell s$ otherwise.
\end{enumerate}
On the other hand, if $r'< \ell$, we distinguish three cases:
\begin{enumerate}[1)]\compact
  \item $r_i > r'$. Then $\lcp(x_i,s)=r'$; thus $x_i \prec_\ell s$ if $\rho^\infty \prec s$, and $x_i\succ_\ell s$ otherwise.
  \item $r_i = r'$. Then we use Lemma~\ref{lem:corr}(\ref{it:pp}) to determine the order between $x_i$ and $s$ trimmed to $\ell$ characters. This, however, may happen only for a single value $i$.
  \item $r_i < r'$. Then $\lcp(x_i,s)=r_i$; thus $x_i \prec_\ell s$ if $x_0 \prec \rho^\infty$, and $x_i\succ_\ell s$ otherwise.
\end{enumerate}
Consequently, in constant time we can partition indices $i$ into at most three ranges,
and for each range determine whether $x_i \prec_\ell s$, $x_i=_\ell s$, or $x_i\succ_\ell s$
for all indices $i$ in the range. We ignore from further computations those ranges for which we already know that $x_i \notin I$,
and for the remaining ones repeat the procedure above with $t$ instead of $s$.
We end up with $\Oh(1)$ ranges of positions $i$ for which $x_i\in I$. However, note that 
as the string sequence $(x_i)_{i=0}^k$ is always monotone (decreasing if $x_0 \preceq \rho^\infty$, increasing otherwise),
these ranges (if any) can be merged into a single range, so in the output we end up with a single (possibly empty) periodic progression.
\mayqed
\end{proof}

\begin{figure}[ht]
\begin{center}

\begin{tikzpicture}[every text node part.0/.style={align=left}]
\tikzset{every node/.style={font={\footnotesize \ttfamily},inner sep=2pt}, inner/.style={draw,rectangle, font=\small, inner sep = 1pt}, leaf/.style={draw,circle, font=\normalsize, minimum size=12pt, inner sep=0pt}}
\begin{scope}[yscale=-.6, xscale=.7]
\node[inner] (root) at (0,0) {1};
\node[leaf] (0) at (-4,1) {2};
\node[inner, fill=black!15] (a) at (-.75,1) {2};
\node[inner] (ab) at (-1.5,2) {4};
\node[inner] (abab) at (-3.5,4) {8};
\node[leaf] (abababb0) at (-5.5, 8) {16};
\node[leaf] (ababb0) at (-3.5, 6) {12};
\node[inner,fill=black!15] (ababbaba) at (-2 ,8) {16};
\node[leaf] (ababbabababb0) at (-2, 13) {26};
\node[leaf] (abb0) at (-1.5, 4) {8};
\node[inner,fill=black!15] (abba) at (-.5, 4) {8};
\node[inner,fill=black!15] (abbababa) at (-.5, 8) {16};
\node[inner] (abbabababb0) at (-.5, 11) {22};
\node[inner] (b) at (3.5,1) {2};
\node[leaf] (b0) at (2,2) {4};
\node[inner] (ba) at (3.5,2) {4};
\node[inner] (baba) at (1,4) {8};
\node[inner,fill=black!15] (babababb) at (1, 8) {16};
\node[leaf] (babababb0) at (1,9) {18};
\node[leaf] (bababb0) at (2, 7) {14};
\node[inner] (babb) at (3.5,4) {8};
\node[leaf] (babb0) at (2.5,5) {10};
\node[inner,fill=black!15] (babbabab) at (3.5,8) {16};
\node[leaf] (babbabababb0) at (3.5,12) {24};
\node[inner] (bb) at (5.5,2) {4};
\node[leaf] (bb0) at (4.5,3) {6};
\node[inner,fill=black!15] (bbab) at (5.5,4) {8};
\node[inner,fill=black!15] (bbababab) at (5.5,8) {16};
\node[leaf] (bbabababb0) at (5.5,10) {20};
\draw (root) --  node[above, pos=.6]{\$} (0) (root) -- node[left, pos=.5]{a} (a) --node[left, pos=.5]{b}  (ab) -- node[left, pos=0.25]{a} node[left, pos=0.75]{b} (abab) -- node[left, pos=0.125]{a} node[left, pos=0.375]{b} node[left, pos=0.625]{b} node[left, pos=0.875]{\$} (abababb0) (abab) -- node[left=-1.5, pos=0.25]{b} node[left=-1.5, pos=0.75]{\$} (ababb0) (abab) -- node[right, pos=0.125]{b} node[right, pos=0.375]{a} node[right, pos=0.625]{b} node[right, pos=0.875]{a} (ababbaba) --  node[right, pos=0.1]{b} node[right, pos=0.3]{a} node[right, pos=0.5]{b} node[right, pos=0.7]{b}  node[right, pos=0.9]{\$} (ababbabababb0) (ab) -- node[left=-1, pos=0.25]{b} node[left=-1, pos=0.75]{\$} (abb0) (ab) --node[right, pos=0.25]{b} node[right, pos=0.75]{a}  (abba) -- node[right, pos=0.125]{b} node[right, pos=0.375]{a} node[right, pos=0.625]{b} node[right, pos=0.875]{a} (abbababa) -- node[right, pos=0.166666]{b} node[right, pos=0.5]{b} node[right, pos=0.8333333]{\$} (abbabababb0) (root) -- node[above, pos=.6]{b} (b) -- node[above, pos=.6]{\$} (b0) (b) -- node[right, pos=.5]{a} (ba) --node[above, pos=.3]{b} node[above, pos=.8]{a}  (baba) -- node[left, pos=0.125]{b} node[left, pos=0.375]{a} node[left, pos=0.625]{b} node[left, pos=0.875]{b} (babababb) -- node[left, pos=.5]{\$} (babababb0) (baba) -- node[right, pos=0.166666]{b} node[right, pos=0.5]{b} node[right, pos=0.8333333]{\$} (bababb0) (ba) -- node[right, pos=0.25]{b} node[right, pos=0.75]{b}  (babb) --  node[above, pos=.8]{\$} (babb0) (babb) -- node[right, pos=0.125]{a} node[right, pos=0.375]{b} node[right, pos=0.625]{a} node[right, pos=0.875]{b} (babbabab) -- node[right, pos=0.125]{a} node[right, pos=0.375]{b} node[right, pos=0.625]{b} node[right, pos=0.875]{\$} (babbabababb0) (b) -- node[above, pos=.6]{b} (bb) -- node[above, pos=.8]{\$} (bb0) (bb) -- node[right, pos=0.25]{a} node[right, pos=0.75]{b} (bbab) --  node[right, pos=0.125]{a} node[right, pos=0.375]{b} node[right, pos=0.625]{a} node[right, pos=0.875]{b} (bbababab) --node[right, pos=0.25]{b} node[right, pos=0.75]{\$}  (bbabababb0);
\end{scope}
\end{tikzpicture}

\vspace{-.5cm}
\end{center}
\caption{An auxiliary tree $T$ introduced to define the wavelet suffix tree of $w=\texttt{ababbabababb}$.
Levels are written inside nodes. Gray nodes are dissolved during the construction of the wavelet suffix tree.}
\label{fig:wst2}
\end{figure}

\subsubsection{Definition of wavelet suffix trees}
\label{ssec:WSTdefinition}
Let $w$ be a string of length $n$. To define the wavelet suffix tree of $w$, we start from an auxiliary tree $T$ of height $\Oh(\log n)$ with $\Oh( n \log n)$ nodes.
Its leaves represent non-empty suffixes of $w\$$, and the left-to-right order of leaves corresponds to the lexicographic order on the suffixes.
Internal nodes of $T$ represent all substrings of $w$ whose length is a power of two, with an exception of the root, which represents
the empty word. Edges in $T$ are defined so that a node representing $v$ is an ancestor of a node representing $v'$
if and only if $v$ is a prefix of $v'$. To each non-root node $u$ we assign a \emph{level} $\ell(u) := 2|v|$,
where $v$ is the substring that $u$ represents. For the root $r$, we set $\ell(r):= 1$.
See Figure~\ref{fig:wst2} for a sample tree $T$ with levels assigned to nodes.

For a node $u$, we define $\S(u)$ to be the set of suffixes of $w\$$ that are represented by descendants of $u$.
Note that $\S(u)$ is a singleton if $u$ is a leaf.
The following observation characterizes the levels and sets $\S(u)$.

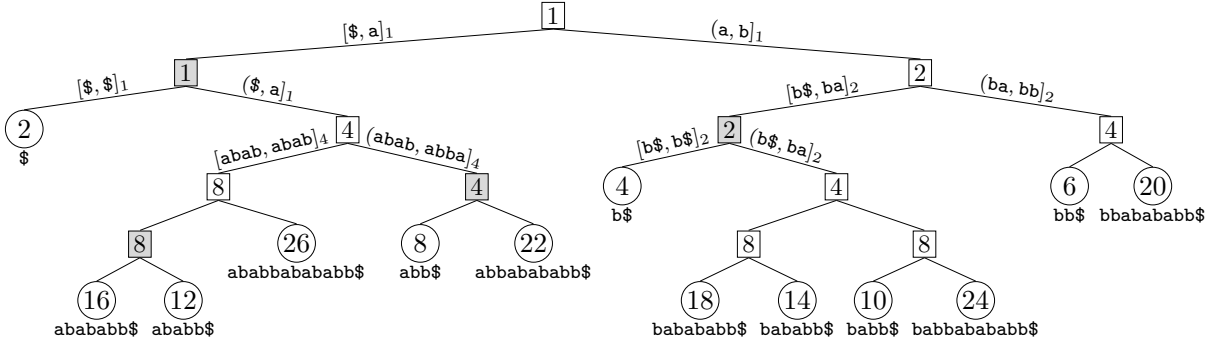
\begin{figure}
\begin{center}

\begin{tikzpicture}[every text node part/.style={align=left},scale=0.95]
\tikzset{every node/.style={font={\scriptsize \ttfamily},inner sep=0.5pt}, every internal node/.style={draw,rectangle, font=\small, inner sep = 2pt}, every leaf node/.style={draw,circle, font=\normalsize, minimum size=15pt},
labl/.style={postaction={decorate,decoration={text along path,reverse path=true, text align=center,raise=3pt, text={{\scriptsize $#1$}{}}}}},
labr/.style={postaction={decorate,decoration={text along path,reverse path=false, text align=center,raise=3pt, text={{\scriptsize $#1$}{}}}}}, level distance=.8cm,sibling distance=0.13cm}
\Tree[.\node (root) {1}; \edge[labl={[\texttt{\$},\texttt{a}]_1}];
 [.\node[fill=black!15]{1};   \edge[labl={[\texttt{\$},\texttt{\$}]_1}]; \node[label=below:\$] {2};
	\edge[labr={(\texttt{\$},\texttt{a}]_1}]; [.4 
		\edge[labl={[\texttt{abab},\texttt{abab}]_4\;\;\;}]; [.{8} 
			  [.\node[fill=black!15] {8}; \node[label=below:abababb\$] {16};   \node[label=below:ababb\$] {12}; ]   \node[label=below:ababbabababb\$] {26};
		]  
		 \edge[labr={\;\;\;(\texttt{abab},\texttt{abba}]_4}]; [ .\node[fill=black!15]{4}; \node[label=below:$\quad$abb\$$\quad$] {8};   \node[label=below:abbabababb\$] {22}; ]
	] 
]  
	 \edge[labr={(\texttt{a},\texttt{b}]_1}]; [ .{2} 
		 \edge[labl={[\texttt{b\$},\texttt{ba}]_2}]; [ .\node[fill=black!15]{2}; \edge[labl={[\texttt{b\$},\texttt{b\$}]_2}];  \node[label=below:b\$] {4};
			 \edge[labr={(\texttt{b\$},\texttt{ba}]_2}]; [ .{4} 
				[ .{8}  \node[label=below:babababb\$] {18}; \node[label=below:bababb\$] {14}; ]
				 [ .{8}  \node[label=below:babb\$] {10};  \node[label=below:babbabababb\$] {24}; ]
			]
		] 
	\edge[labr={(\texttt{ba},\texttt{bb}]_2}]; [ .{4} \node[label=below:bb\$] {6};  \node[label=below:bbabababb\$] {20}; ]
	] 
];
\end{tikzpicture}
\vspace{-.5cm}
\end{center}
\caption{The wavelet suffix tree of $w=\texttt{ababbabababb}$ (see also Figures~\ref{fig:wst} and~\ref{fig:wst2}). 
Levels are written inside nodes. Gray nodes have been introduced as inner nodes of replacement trees.
The corresponding suffix is written down below each leaf. Selected edges $e$ are labelled with the intervals $I(e)$.
}
\label{fig:wst3}
\end{figure}

\begin{observation}\label{obs:ell}
For any node $u$ other than the root:
\begin{enumerate}[(1)]\compact
  \item\label{it:mon} $\ell(\parent(u))\le \ell(u)$,
  \item\label{it:here} if $y\in \S(u)$ and $y'$ is a suffix of $w\$$ such that $\lcp(y,y')\ge \ell(\parent(u))$, then $y'\in \S(u)$,
  \item\label{it:sim} if $y,y'\in \S(u)$, then $\lcp(y,y')\ge \floor{\frac{1}{2}\ell(u)}$.
\end{enumerate}
\end{observation}
\pagebreak

Next, we modify $T$ to obtain a binary tree of $\Oh(n)$ nodes.
In order to reduce the number of nodes, we dissolve all nodes with exactly one child, i.e., while there is a non-root node $u$ with exactly one child $u'$,
we set $\parent(u'):=\parent(u)$ and remove $u$.
To make the tree binary, for each node $u$ with $k>2$ children, we remove the edges 
between $u$ and its children, and introduce a \emph{replacement tree},  a full binary tree
with $k$ leaves whose root is $u$, and leaves are the $k$ children of $u$
(preserving the left-to-right order).
We choose the replacement trees, so that the resulting tree still has depth $\Oh(\log n)$.
In Section~\ref{ssec:ttools} we provide a constructive proof that such a choice is possible.
This procedure introduces new nodes (inner nodes of the replacement trees); their levels are inherited
from the parents.

The obtained tree is the wavelet suffix tree of $w$; see Figure~\ref{fig:wst3}
for an example.
Observe that, as claimed in Section~\ref{ssec:over} it is a full binary tree of logarithmic height,
whose leaves corresponds to non-empty suffixes of $w\$$.
Moreover, it is not hard to see that this tree still satisfies Observation~\ref{obs:ell}.

As described in Section~\ref{ssec:over}, each node of $u$ (except for the leaves)
stores two bitmasks. In either bitmask each bit corresponds to a suffix $y\in \S(u)$, and it is equal to 0 if $y\in \S(\lchild(u))$
and to 0 if $y\in \S(\rchild(u))$, where $\lchild(u)$ and $\rchild(u)$ denote the children of $u$.
In the first bitmask the suffixes $y=w[j..]\$$ are ordered by the starting position $j$,
and in the second bitmask~--- by pairs $(w[j-1], j)$ (assuming $w[-1]=\$$).
Both bitmasks are equipped with rank/select data structures.

Additionally, each node and each edge of the wavelet suffix tree are associated 
with a string interval whose endpoints are suffixes of $w\$$.
Namely, for an arbitrary node $u$ we define $I(u)=[\min \S(u), \max\S(u)]_{\ell(u)}$.
Additionally, if $u$ is not a leaf, we set
$I(u, \lchild(u)) = [\min \S(u), y]_{\ell(u)}$
and 
$I(u, \rchild(u)) = (y, \max\S(u)]_{\ell(u)}$,
where $y=\max\S(\lchild(u))$ is the suffix corresponding to
the rightmost leaf in the left subtree of $u$; see also Figure~\ref{fig:wst3}. 
For each node $u$ we store the starting positions of $\min\S(u)$ and $\max\S(u)$
in order to efficiently retrieve a representation of $I(u)$ and $I(e)$ for adjacent edges $e$.
The following lemma characterizes the intervals.

\begin{lemma}\label{lem:intervals}
For any node $u$ we have:
\begin{enumerate}[(1)]\compact
  \item\label{it:disj} If $u$ is not a leaf, then $I(u)$ is a disjoint union of $I(u, \lchild (u))$ and $I(u, \rchild(u))$.
  \item\label{it:iff} If $y$ is a suffix of $w\$$, then $y\in I(u)$ if and only if $y\in \S(u)$.
   \item\label{it:sub} If $u$ is not the root, then $I(u)\sub I(\parent(u),u)$. 
\end{enumerate}
\end{lemma}
\begin{proof}
\noindent (\ref{it:disj}) is a trivial consequence of the definitions.

\smallskip
\noindent (\ref{it:iff}) Clearly $y\in \S(u)$ iff $y\in [\min \S(u),\max \S(u)]$.
Therefore, it suffices to show that if $\lcp(y, y')\ge \ell(u)$ for $y'=\min\S(u)$ or $y'=\max\S(u)$, then $y\in \S(u)$.
This is, however, a consequence of points (\ref{it:mon}) and (\ref{it:here}) of Observation~\ref{obs:ell}.

\smallskip
\noindent (\ref{it:sub}) Let $\ell_p=\ell(\parent(u))$. 
If $u=\lchild(\parent(u))$, then $\S(u)\sub \S(\parent(u))$ and, by Observation~\ref{obs:ell}(\ref{it:mon}), $\ell(u)\le \ell_p$, which implies the statement. 

Therefore, assume that $u=\rchild(\parent(u))$,
and let $u'$ be the left sibling of $u$. Note that $I(\parent(u), u)=(\max \S(u'), \max\S(u)]_{\ell_p}$
and $I(u)\sub [\min \S(u),\max \S(u)]_{\ell_p}$, since $\ell(u)\le \ell_p$.
Consequently, it suffices to prove that $\max \S(u') \prec_{\ell_p} \min \S(u)$.
This is, however, a consequence of Observation~\ref{obs:ell}(\ref{it:here}) for $y=\min \S(u)$ and $y'=\max\S(u')$,
and the fact that the left-to-right order of leaves coincides with the lexicographic order of the corresponding
suffixes of $w\$$.
\mayqed \end{proof}

For each edge $e=(\parent(u),u)$ of the wavelet suffix tree,
we define $L(e)$ to be the sorted list of those substrings of $w$ which belong to $I(e) \setminus I(u)$.

Recall that the wavelet suffix tree shall enjoy the \emph{lexicographic property}:
if we traverse the tree, and when going \emph{down} an edge $e$ we write out the contents of $L(e)$, whereas when visiting a leaf we output the corresponding suffix of $w\$$, we shall obtain a lexicographically sorted list of all substrings of $w\$$. 
This is proved in the following series of lemmas.

\begin{lemma}\label{lem:prelex}
Let $e=(\parent(u),u)$ for a node $u$. 
Substrings in $L(e)$ are smaller than any string in~$I(u)$.
\end{lemma}
\begin{proof}
We use a shorthand $\ell_p$ for $\ell(parent(u))$. 
Let $y=\max \S(u)$ be the rightmost suffix in the subtree of~$u$. Consider a substring $s = w[k..j] \in L(e)$,
also let $t=w[k..]\$$. 

We first prove that $s\preceq y$. Note that $I(e)=[x,y]_{\ell_p}$ or $I(e)=(x,y]_{\ell_p}$ for some string $x$.
We have $s\in L(e)\sub I(e)$, and thus $s\preceq_{\ell_p} y$.  If $\lcp(s,y)<\ell_p$, this already implies that $s\preceq y$.
Thus, let us assume that $\lcp(s,y)\ge \ell_p$. The suffix $t$ has $s$ as a prefix, so this also means that $\lcp(t,y)\ge \ell_p$.
By Observation~\ref{obs:ell}(\ref{it:here}), $t\in \S(u)$, so $t\preceq y$.
Thus $s\preceq t \preceq y$, as claimed.

To prove that $s\preceq y$ implies that $s$ is smaller than any string in $I(u)$,
it suffices to note that $y\in \S(u)\sub I(u)$, $s\notin I(u)$, and $I(u)$ is an interval. 
\mayqed
\end{proof}

\begin{lemma}\label{lem:lex}
The wavelet suffix tree satisfies the lexicographic property.
\end{lemma}
\begin{proof}
Note that for the root $r$ we have $I(r)=[\$, c]_1$ where $c$ is the largest character present in $w$.
Thus, $I(r)$ contains all substrings of $w\$$ and it suffices to show that  if we traverse the subtree of $r$,
writing out the contents of $L(e)$ when going down an edge $e$, and the corresponding suffix when visiting a leaf,
we obtain a sorted list of substrings of $w\$$ contained in $I(r)$. But we will show even a stronger claim~--- we will show that, in fact, this property holds for all nodes $u$ of the tree.

If $u$ is a leaf this is clear, since $I(u)$ consists of the corresponding suffix of $w\$$ only.
Next, if we have already proved the hypothesis for $u$, then prepending the output with the contents of $L(\parent(u),u)$,
by Lemmas~\ref{lem:prelex} and~\ref{lem:intervals}(\ref{it:sub}), we obtain a sorted list of substrings of $w\$$ contained in $I(\parent(u),u)$.
Applying this property for both children of a non-leaf $u'$, we conclude that if the hypothesis 
holds for children of $u'$ then, by Lemma~\ref{lem:intervals}(\ref{it:disj}), it also holds for $u'$.
\mayqed
\end{proof}

\begin{corollary}\label{cor:cons_pref}
Each list $L(e)$ contains consecutive prefixes of the largest element of $L(e)$.
\end{corollary}
\begin{proof}
Note that if $x\prec y$ are substrings of $w$ such that $x$ is not a prefix of $y$,
then $x$ can be extended to a suffix $x'$ of $w\$$ such that $x\prec x' \prec y$.
However, $L(e)$ does not contain any suffix of~$w\$$.
By Lemma~\ref{lem:lex}, $L(e)$ contains a consecutive collection of substrings of $w\$$,
so $x$ and $y$ cannot be both present in $L(e)$.
Consequently, each element of $L(e)$ is a prefix of $\max L(e)$.

Similarly, since $L(e)$ contains a consecutive collection of substrings of $w\$$,
it must contain all prefixes of $\max L(e)$ no shorter than $\min L(e)$.
\mayqed \end{proof}

\subsubsection{Implementation of auxiliary queries}\label{ssec:ops}
Recall that $L_x(e)$ is the sublist of $L(e)$ containing suffixes of $x$.
The wavelet suffix tree shall allow the following four types of queries in constant time:
\begin{enumerate}[(1)]\compact
  \item\label{q:lxe} For a substring $x$ and an edge $e$, output the list $L_x (e)$ represented as $\Oh(1)$ non-overlapping periodic progressions;
  \item\label{q:cnt} Count the number of suffixes of $x = w[i..j]$ in the left/right subtree of a node (given along with the segment of its first bitmask  corresponding to suffixes that start inside $[i, j]$);
  \item\label{q:ccnt} Count the number of suffixes $x=w[i..j]$ that are preceded by a character~$c$ and lie in the left/right subtree of a node (given along with the segment of its second bitmask corresponding to suffixes that start inside $[i, j]$ and are preceded by $c$);
  \item\label{q:rle} For a substring $x$ and and edge $e$, compute the run-length encoding of the sequence of characters preceding suffixes in $L_x(e)$.
\end{enumerate}

We start with an auxiliary lemma applied in the solutions to all four queries.
\begin{lemma}\label{lem:aux}
Let $e=(u,u')$ be an edge of a wavelet suffix tree of $w$, with $u'$ being a child of $u$.
The following operations can be implemented in constant time.
\begin{enumerate}[(1)]\compact
  \item\label{it:subs} Given a substring $x$ of $w$, $|x|< \ell(u)$, return, as a single periodic progression of starting positions, all suffixes $s$ of $x$ such that $s\in I(e)$.
  \item\label{it:range} Given a range of positions $[i,j]$, $j-i\le \ell(u)$, return all positions $k\in [i,j]$ such that $w[k..]\$ \in I(e)$,
  represented as at most two non-overlapping periodic progressions.
\end{enumerate}
\end{lemma}
\begin{proof}
Let $p$ be the longest common prefix of all strings in $I(u)$; by Observation~\ref{obs:ell}(\ref{it:sim}), we have $|p|\ge \lfloor\frac{1}{2}\ell(u)\rfloor$.

\noindent
(\ref{it:subs}) Assume $x=w[i..j]$. We apply Theorem~\ref{th:occurrences} to find all occurrences of $p$ within $x$,
represented as a single periodic progression since $|x|+1< 2(|p|+1)$. Then, using Lemma~\ref{lem:IntervalSelection}, we filter positions $k$
for which $w[k,j]\in I(e)$.

\noindent
(\ref{it:range}) Let $x=w[i..j+|p|-1]$ ($x=w[i..]\$$ if $j+|p|-1>|w|$). We apply Theorem~\ref{th:occurrences} to find all occurrences of $p$ within~$x$,
represented as at most two periodic progressions since $|x|+1\le \ell(u)+|p|+1\le 2\lfloor\frac12\ell(u)\rfloor+|p|+2<3(|p|+1)$.
Like previously, using Lemma~\ref{lem:IntervalSelection} we filter positions $k$
for which $w[k..]\$\in I(e)$.
\mayqed \end{proof}

\begin{lemma}\label{lem:sxe}
The wavelet suffix tree allows to answer queries (\ref{q:lxe}) in constant time. In more details, for an edge $e=(\parent(u), u)$, the starting positions of suffixes in $L_x(e)$ form at most three non-overlapping periodic progressions, which can be reported in $\Oh(1)$ time.
\end{lemma}
\begin{proof}
First, we consider short suffixes. 
We use Lemma~\ref{lem:aux}(\ref{it:subs}) to find all suffixes $s$ of~$x$, $|s|<\ell(\parent(u))$,
such that $s\in I(\parent(u),u)$. Then, we apply Lemma~\ref{lem:IntervalSelection} to filter 
out all suffixes belonging to $I(u)$. By Lemma~\ref{lem:prelex}, we obtain
at most one periodic progression. 

Now, it suffices to generate suffixes $s$, $|s| \ge \ell(\parent(u))$, that belong to $L(e)$. Suppose
$s=w[k..j]$. If $s\in I(e)$, then equivalently $w[k..]\$\in I(e)$,
since $s$ is a long enough prefix of $w[k..]\$$ to determine
whether the latter belongs to $I(e)$.
Consequently, by Lemma~\ref{lem:intervals}, $w[k..]\$\in I(u)$.
This implies $|s|<\ell$ (otherwise we would have $s\in I(u)$), i.e., $k\in [j-\ell+2..j-\ell(\parent(u))+1]$.
We apply Lemma~\ref{lem:aux}(\ref{it:range}) to compute all positions $k$ in this range for which $w[k..]\$\in I(e)$.
Then, using Lemma~\ref{lem:IntervalSelection}, we filter out positions $k$ such that $w[k..j]\in I(u)$.
By Lemma~\ref{lem:prelex}, this cannot increase the number of periodic progressions, 
so we end up with three non-overlapping periodic progressions in total. 
\mayqed
\end{proof}

\begin{lemma}\label{lem:i}
The wavelet suffix tree allows to answer queries (\ref{q:cnt}) in constant time.
\end{lemma}
\begin{proof}
Let $u$ be the given node and $u'$ be its right/left child (depending on the variant of the query).
First, we use Lemma~\ref{lem:aux}(\ref{it:subs}) to find all suffixes $s$ of~$x$, $|s|<\ell(u)$,
such that $s\in I(u,u')$, i.e., $s$ lies in the appropriate subtree of $u$.

Thus, it remains to count suffixes of length at least $\ell(u)$.
Suppose $s=w[k..j]$ is a suffix of $x$ such that $|s|\ge \ell(u)$ and $s\in I(u,u')$.
Then $w[k..]\$\in I(u,u')$, and the number of suffixes $w[k..]\$\in I(u,u')$ such that $k \in [i, j]$
is simply the number of 1's or 0's in the given segment of the first bitmask in~$u$,
which we can compute in constant time.
Observe, however, that we have also counted positions $k$ such that $|w[k..j]|<\ell(u)$,
and we need to subtract the number of these positions.
For this, we use Lemma~\ref{lem:aux}(\ref{it:range}) to compute
the positions $k\in [j-\ell+2, j]$ such that $w[k..]\$\in I(u,u')$.
We count the total size of the obtained periodic progressions,
and subtract it from the final result, as described.
\mayqed
\end{proof}

\begin{lemma}
The wavelet suffix tree allows to answer queries (\ref{q:ccnt}) and (\ref{q:rle}) in $\Oh(1)$ time.
\end{lemma}
\begin{proof}
Observe that for any periodic progression $p_0\ldots,p_k$ we have $w[p_1-1]=\ldots=w[p_k-1]$.
Thus, it is straightforward to determine which positions of such a progression are preceded by $c$.

Answering queries (\ref{q:ccnt}) is analogous to answering queries (\ref{q:cnt}), we just use the second bitmask at the given node and consider only positions preceded by $c$ while counting the sizes of periodic progressions.

To answer queries (\ref{q:rle}), it suffices to use Lemma~\ref{lem:sxe} to obtain $L_x(e)$.
By Corollary~\ref{cor:cons_pref}, suffixes in $L_x(e)$ are prefixes of one another,
so the lexicographic order on these suffixes coincides with the order of ascending lengths. 
Consequently, the run-length encoding of the piece corresponding to $L_x(e)$ has at most
six phrases and can be easily found in $\Oh(1)$ time.
\mayqed
\end{proof}

\subsection{Construction of wavelet suffix trees}\label{ssec:constr}
The actual construction algorithm is presented in Section~\ref{sec:constr}.
Before, in Section~\ref{ssec:ttools}, we introduce several
auxiliary tools for abstract weighted trees.

\subsubsection{Toolbox for weighted trees}\label{ssec:ttools}

Let $T$ be a rooted ordered tree with positive integer weights on edges, $n$ leaves and no inner nodes of degree one.
We say that $L_1,\ldots,L_{n-1}$ is an \emph{LCA sequence} of $T$,
if $L_i$ is the (weighted) depth of the lowest common ancestor
of the $i$-th and $(i+1)$-th leaves.
The following fact is usually applied to construct the suffix tree of a string from the suffix array and the LCP table~\cite{AlgorithmsOnStrings}.

\begin{fact}\label{fct:lca}
Given a sequence $(L_i)_{i=1}^{n-1}$ of non-negative integers,
one can construct in $\Oh(n)$ time a tree whose LCA sequence is $(L_i)_{i=1}^{n-1}$.
\end{fact}
The LCA sequence suffices to detect if a tree is binary.
\begin{observation}\label{obs:lca}
A tree is a binary tree if and only if its LCA sequence $(L_i)_{i=1}^{n-1}$ satisfies the following property
for every $i<j$: if $L_i = L_j$ then there exists $k$, $i<k<j$, such that $L_k<L_i$.
\end{observation}

Trees constructed by the following lemma can be
seen as a variant of the weight-balanced trees, whose existence for arbitrary weights was by proved Blum and Mehlhorn \cite{DBLP:journals/tcs/BlumM80}.

\begin{lemma}\label{lem:wb}
Given a sequence $w_1,\ldots,w_n$ of positive integers,
one can construct in $\Oh(n)$ time a binary tree $T$ with $n$ leaves,
such that the depth of the $i$-th leaf is $\Oh(1+\log \frac{W}{w_i})$, where $W =\sum_{j=1}^{n} w_j$.
\end{lemma}
\begin{proof}
For $i=0,\ldots,n$ define $W_i = \sum_{j=1}^i w_j$.
Let $p_i$ be the position of the most significant bit where the binary representations of $W_{i-1}$ and $W_i$ differ,
and let $P=\max_{i=1}^n p_i$.
Observe that $P=\Theta(\log W)$ and $p_i = \Omega(\log w_i)$.
Using Fact~\ref{fct:lca}, we construct a tree $T$ whose LCA sequence is $L_i=P-p_i$. Note
that this sequence satisfies the condition of Observation~\ref{obs:lca}, and thus the tree is binary.

Next, we insert an extra leaf between the two children of any node to make the tree ternary.
The $i$-th of these leaves is inserted at (weighted) depth $1+L_i = \Oh(1+\log W-\log w_i)$,
which is also an upper bound for its unweighted depth. Next, we remove the original leaves.
This way we get a tree satisfying the lemma, except for the fact that inner nodes may have between one and three children,
rather than exactly two.
In order to resolve this issue, we remove (dissolve) all inner nodes with exactly one child, and for each node $u$ with three children $v_1, v_2, v_3$,
we introduce a new node $u'$, setting $v_1, v_2$ as the children of $u'$ and $u', v_3$ as the children of $u$.
This way we get a full binary tree, and the depth of any node may increase at most twice,
i.e., for the $i$-th leaf it stays~$\Oh(1+\log\frac{W}{w_i})$.
\mayqed \end{proof}

Let $T$ be an ordered rooted tree and let $u$ be a node of $T$, which is neither the root nor a leaf.
Also, let $v$ be the parent of $u$. We say that $T'$ is obtained from $T$ by \emph{contracting} the edge $(v,u)$, 
if $u$ is removed and the children of $u$ replace $u$ at its original location in the list of children of $v$.
If $T'$ is obtained from $T$ by a sequence of edge contractions, we say that $T'$ is a \emph{contraction} of $T$. 
Note that contraction does not alter the pre-order and post-order of the preserved nodes, which implies
that the ancestor-descendant relation also remains unchanged for these nodes.

\begin{corollary}
\label{cor:wb}
Let $T$ be an ordered rooted tree of height $h$, which has $n$ leaves and no inner node with exactly one child.
Then, in $\Oh(n)$ time one can construct a full binary ordered rooted tree $T'$ of height $\Oh(h+\log n)$  such that
$T$ is a contraction of $T'$ and
$T'$ has $\Oh(n)$ nodes.
\end{corollary}
\begin{proof}
For any node $u$ of $T$ with three or more children, we replace the star-shaped tree joining it with its children $v_1,\ldots,v_k$
by an appropriate replacement tree. Let $W(u)$ be the number of leaves in the subtree of $u$, and let $W(v_i)$ be the number of leaves in the subtrees below $v_i$, $1\le i \le k$. We use Lemma~\ref{lem:wb} for $w_i = W(v_i)$ to construct the replacement tree. Consequently, a node $u$ with depth $d$ in $T$ has depth $\Oh(d+\log\frac{n}{W(u)})$ in $T'$ (easy top-down induction). The resulting tree has height $\Oh(h+\log n)$, as claimed.
\mayqed \end{proof}

\subsubsection{The algorithm}\label{sec:constr}
In this section we show how to construct the wavelet suffix tree of a string $w$ of length $n$ in $\Oh(n\sqrt{\log n})$ time.
The algorithm is deterministic, but the data structure of Theorem~\ref{th:occurrences},
required by the wavelet suffix tree, has a randomized construction only, running in $\Oh(n)$ expected time.

The construction algorithm has two phases: first, it builds the \emph{shape}
of the wavelet suffix tree, following a description in Section~\ref{ssec:WSTdefinition},
and then it uses the results of Section~\ref{sec:wt} to obtain the bitmasks.

We start by constructing the suffix array and the LCP table for $w\$$ (see~\cite{AlgorithmsOnStrings}). 
Under the assumption that $\sigma < n$, this takes linear time.

Recall that in the definition of the wavelet suffix tree we started with a tree of size $\Oh(n\log n)$.
For an $o(n\log n)$-time construction we cannot afford that. Thus, we construct the tree $T$ already without inner nodes
having exactly one child.
Observe that this tree is closely related to the suffix tree of $w\$$. The only difference
is that if the longest common prefix of two consecutive suffixes is~$d$, their root-to-leaf paths diverge at depth $\floor{\log d}$ instead of $d$.
To overcome this difficulty, we use Fact~\ref{fct:lca} for $L_i = \floor{\log LCP[i]}$, rather than $LCP[i]$ which we 
would use for the suffix tree. This way an inner node $u$ at depth
$j$ represents a substring of length $2^j$. The level $\ell(u)$ of an inner node $u$ is set to $2^{j+1}$, and if $u$ is a leaf representing a suffix $s$ of $w\$$, we have $\ell(u)=2|s|$.

After this operation, the tree $T$ may have inner nodes of large degree,
so we use Corollary~\ref{cor:wb} to obtain a binary tree such that $T$ is its contraction. 
We set this binary tree as the shape of the wavelet suffix tree. 
Since $T$ has height $\Oh(\log n)$, so does the wavelet suffix tree.

To construct the bitmasks, we apply Theorem~\ref{thm:aswt} for $T$ with the leaf representing
$w[i..]\$$ assigned to $i$.
For the first bitmask, we simply set $s[i]=i$.
For the second bitmask,
we sort all positions $i$ with respect to $(w[i-1],i)$ and take the resulting sequence as $s$.

This way, we complete the proof of the main theorem concerning wavelet suffix trees.
\begin{theorem}\label{th:wst}
A wavelet suffix tree of a string $w$ of length $n$ occupies $\Oh(n)$ space and can be constructed in $\Oh(n \sqrt{\log{n}})$ expected time.
\end{theorem}

\subsection{Applications}\label{ssec:apps}
\subsubsection{Substring suffix rank/select}
In the substring suffix rank problem, we are asked to find the rank of a substring $y$ among the suffixes of another substring $x$. 
The substring suffix selection problem, in contrast, 
is to find the $k$-th lexicographically smallest suffix of $x$ for a given an integer $k$ and a substring $x$ of $w$. 

\begin{theorem}
\label{thm:ssrank}
The wavelet suffix tree can solve the substring suffix rank problem in $\Oh(\log n)$ time.
\end{theorem}
\begin{proof}
Using binary search on the leaves of the wavelet suffix tree of $w$, we locate the minimal suffix $t$ of $w\$$ such that $t\succ y$. 
Let $\pi$ denote the path from the root to the leaf corresponding to $t$.
Due to the lexicographic property, the rank of $y$ among the suffixes of $x$ is equal to the sum of two numbers.
The first one is the number of suffixes of $x$ in the left subtrees hanging from the path $\pi$,
whereas the second summand is the number of suffixes not exceeding $y$ in the lists $L_x(e)$ for $e\in \pi$.

To compute those two numbers, we traverse $\pi$ maintaining a segment $[\ell,r]$ of the first bitmask corresponding
to the suffixes of $w\$$ starting within $x$.
When we descend to the left child, we set $[\ell, r]: = [rank_0(\ell), rank_0(r)]$, while for the right child, we set $[\ell, r]: = [rank_1(\ell), rank_1(r)]$.
In the latter case, we pass $[\ell,r]$ to type (\ref{q:cnt}) queries, which let us count the suffixes of
$x$ in the left subtree hanging from $\pi$ in the current node.
This way, we compute the first summand.

For the second number, we use type (\ref{q:lxe}) queries to generate all lists $L_x(e)$ for $e\in \pi$.
Note that if we concatenated these lists $L_x(e)$ in the root-to-leaf order of edges,
we would obtain a sorted list of strings.
Thus, while processing the lists in this order (ignoring the empty ones), we add up the sizes of $L_x(e)$ until $\max L_x(e) \succ y$.
For the first encountered list $L_x(e)$ satisfying this property, we binary search within $L_x(e)$
to determine the number of elements not exceeding $y$, and also add this value to the final result.

The described procedure requires $\Oh(\log n)$ time,
since type (\ref{q:lxe}) and (\ref{q:cnt}) queries, as well as substring comparison queries (Lemma~\ref{lem:all}), run in $\Oh(1)$ time.
\mayqed
\end{proof}

\begin{theorem}
\label{thm:ssselect}
The wavelet suffix tree can solve the substring suffix selection problem in $\Oh(\log n)$ time.
\end{theorem}
\begin{proof}
The algorithm traverses a path in the wavelet suffix tree of $w$. 
It maintains a segment $[\ell,r]$ of the first bitmask corresponding to suffixes of $w$ starting within $x=w[i..j]$, 
and a variable $k'$ counting the suffixes of $x$ represented in the left subtrees hanging from the path
on the edges of the path.
The algorithm starts at the root initializing $[\ell, r]$ with $[i,j]$ and $k'=0$. 

At each node $u$, it first decides to which child of $u$ to proceed.
For this, is performs a type (\ref{q:cnt}) query to determine $k''$, the number of suffixes of $x$ 
in the left subtree of $u$. If $k'+k''\ge k$, it chooses to go to the left child, otherwise to the right one;
in the latter case it also updates $k':= k'+k''$.
The algorithm additionally updates the segment $[\ell,r]$ using the rank queries on the bitmask.

Let $u'$ be the child of $u$ that the algorithm has chosen to proceed to.
Before reaching $u'$, the algorithm performs a type (\ref{q:lxe}) query to compute $L_x(u,u')$.
If $k'$ summed with the size of this list is at least $k$, then the algorithm
terminates, returning the $k-k'$-th element of the list (which is easy to retrieve from the representation 
as a periodic progression).
Otherwise, it sets $k':= k'+|L_x(u,u')|$, so that $k'$ satisfies the definition for the extended path from the root to $u'$.

The correctness of the algorithm follows from the lexicographic property,
which implies that at the beginning of each step, the sought suffix of $x$ is the $k-k'$-th
smallest suffix in the subtree of~$u$.
In particular, the procedure always terminates before reaching a leaf.
The running time of the algorithm is $\Oh(\log n)$ due to $\Oh(1)$-time implementations
of type (\ref{q:lxe}) and (\ref{q:cnt}) queries.
\mayqed
\end{proof}

We now show that the query time for the two considered problems is almost optimal. We start by reminding lower bounds by P\u{a}tra\c{s}cu and by J{\o}rgensen and Larsen.

\begin{theorem}[\cite{DBLP:journals/siamcomp/Patrascu11,DBLP:conf/stoc/Patrascu07}]
In the cell-probe model with $W$-bit cells, a static data structure
of size $c \cdot n$ must take $\Omega(\frac{\log n}{\log c + \log W})$ time for  orthogonal range counting queries.
\end{theorem}

\begin{theorem}[\cite{DBLP:conf/soda/JorgensenL11}]
In the cell-probe model with $W$-bit cells, a static data structure
of size $c \cdot n$ must take $\Omega(\frac{\log n}{\log c + \log W})$ time for  orthogonal range selection queries.
\end{theorem}

Both of these results allow the coordinates of points to be in the \emph{rank space}, i.e.,
for each $i\in \{1,\ldots,n\}$ there is one point $(i,A[i])$,
and values $A[i]$ are distinct integers in $\{1,\ldots,n\}$.

This lets us construct a string $w=A[1]\ldots A[n]$ for any given point set $P$. Since $w$ has pairwise distinct characters, comparing suffixes of any substring of $w$ is equivalent to
comparing their first characters.  Consequently, the substring suffix selection in $w$ is equivalent to the orthogonal range selection in~$P$, and the substring suffix rank in $w$ is equivalent to the orthogonal range counting in $P$ (we need to subtract
the results of two suffix rank queries to answer an orthogonal range counting query). Consequently, we obtain
\begin{corollary}\label{cor:lb}
In the cell-probe model with $W$-bit cells, a static data structure
of size $c \cdot n$ must take $\Omega(\frac{\log n}{\log c + \log W})$ time for the substring suffix rank and the substring suffix select queries.
\end{corollary}

\subsubsection{Run-length encoding of the BWT of a substring} 
\label{sec:BWTruns}
\renewcommand{\b}{\mathrm{b}}
Wavelet suffix trees can be also used to compute the run-length encoding of the Burrows-Wheeler transform of a substring. We start by reminding the definitions.
The \emph{Burrows-Wheeler transform}~\cite{BWT} (BWT) of a string $x$ is a string $b_0 b_1 \ldots b_{|x|}$, where $b_k$ is the character preceding the $k$-th lexicographically smallest suffix of $x\$$. The BWT tends to contain long segments of equal characters,
called \emph{runs}. This, combined with \emph{run-length encoding}, allows to compress strings efficiently.
The \emph{run-length encoding} of a string is obtained by replacing each maximal run by a pair: the character that forms the run and the length of the run.  
For example, the BWT of a string $banana$ is $annb\$aa$, and the run-length encoding of $annb\$aa$ is $a1n2b1\$1a2$. 

Below, we show how to compute the run-length encoding of the BWT of a substring $x = w[i..j]$ using the wavelet suffix tree of $w$.
Let $x=w[i..j]$ and for $k\in\{1,\ldots,|x|\}$ let $s_k = w[i_k..j]$ be the suffixes of $x$ sorted in the lexicographic order. Then the BWT of $x$ is equal to $b_0 b_1 \ldots b_{|x|}$, where $b_0 = w[j]$,  and for $k\ge 1$, $b_{k} = w[i_k-1]$, unless $i_k = i$ when $b_{k} = \$$.

Our algorithm initially generates a string $\b(x)$ which instead of $ \$$ contains $w[i-1]$. However, we know that $ \$$ should occur at the position equal to the rank of $x$ among all the suffixes of $x$. Consequently, a single substring suffix rank query suffices to find the position which needs to be corrected.

Remember that the wavelet suffix tree satisfies the lexicographic property. Consequently, if we traverse the tree and write out the characters preceding the suffixes in the lists $L_x(e)$, we obtain $\b(x)$ (without the first symbol $b_0$).
Our algorithm simulates such a traversal.
Assume that the last character appended to $\b(x)$ is $c$, and the algorithm is to move down an edge $e=(u,u')$.
Before deciding to do so, it checks whether all the suffixes of $x$ in the appropriate (left or right) subtree of $u$ are preceded with~$c$.
For this, it performs type (\ref{q:cnt}) and (\ref{q:ccnt}) queries, and if both results are equal to some value $q$,
it simply appends $c^q$ to $\b(x)$ and decides not to proceed to $u'$.
In order to make these queries possible, for each node on the path from the root to $u$,
the algorithm maintains segments corresponding to $[i,j]$ in the first bitmasks, and to $(c, [i,j])$ in the second bitmasks.
These segments are computed using rank queries on the bitmasks while moving down the tree.

Before the algorithm continues at $u'$, if it decides to do so, suffixes in $L_x(e)$ need to be handled.
We perform a type (\ref{q:rle}) query to compute the characters preceding these suffixes, and append
the result to $\b(x)$.
This, however, may result in $c$ no longer being the last symbol appended to $\b(x)$.
If so, the algorithm updates the segments of the second bitmask for all nodes on the path from the root to $u'$.
We assume that the root stores all positions $i$ sorted by $(w[i-1],i)$, which
lets us use a binary search to find either endpoint of the segment for the root. 
For the subsequent nodes on the path, the the rank structures on the second bitmasks are applied. 
Overall, this update  takes $\Oh(\log n)$ time and it is necessary at most once per run.

Now, let us estimate the number of edges visited. Observe that if we go down an edge, then the last character of $\b(x)$ changes before we go up this edge. Thus, all the edges traversed down between such character changes form a path.
The length of any path is $\Oh(\log n)$, and consequently the total number of visited edges is $\Oh(s \log n)$,
where $s$ is the number of runs.

\begin{theorem}
The wavelet suffix tree can compute the run-length encoding of the BWT of a substring $x$  in 
$\Oh(s \log n)$ time, where $s$ is the size of the encoding.
\end{theorem}

\subsubsection{Speeding up queries}\label{sec:speedup}
Finally, we note that building wavelet suffix trees for several
substrings of $w$, we can make the query time adaptive to the length of the query substring $x$,
i.e., replace $\Oh(\log n)$ by $\Oh(\log |x|)$.

\begin{theorem}
\label{thm:speedup}
Using a data structure of size $\Oh(n)$, which can be constructed in $\Oh(n \sqrt{\log n})$ expected time, substring suffix rank and selection problems can be solved in $\Oh(\log |x|)$ time. The run-length encoding $b(x)$ of the BWT of a substring $x$ can be found in 
$\Oh(|b(x)|\log |x|)$ time. 
\end{theorem}
\begin{proof}
We build wavelet suffix trees for some substrings of length $n_k=\lfloor{n^{2^{-k}}}\rfloor$,
$0\le k \le \log\log n$. For each length $n_k$ we choose every $\lfloor \frac{1}{2}n_k\rfloor$-th substring, starting from the prefix
and, additionally, we choose the suffix.
Auxiliary data structures of Lemma~\ref{lem:all} and Theorem~\ref{th:occurrences}, are built for $w$ only.

We have $n_{k}=\floor{\sqrt{n_{k-1}}}$, so $n_{k-1}\le (n_{k}+1)^2$ and thus any substring $x$ of $w$ lies within a substring $v$, $|v|\le 2(|x|+1)^2$, for which we store the wavelet suffix tree.
For each $m$, $1\le m \le n$, we store such $n_k$ that $2m\le n_k \le 2(m+1)^2$. 
This reduces finding an appropriate substring $v$ to simple arithmetics.
 Using the wavelet suffix tree for $v$ instead of the tree for the whole string $w$ gives the announced query times.
 The only thing we must be careful about is that the input for the substring suffix rank problem also consists of a string $y$, which does not need to be a substring of $v$. 
 However, looking at the query algorithm, it is easy to see that $y$ is only used through the data structure of Lemma~\ref{lem:all}.

It remains to analyze the space usage and construction time.
Observe that the wavelet suffix tree of a substring $v$ is simply a binary tree with two bitmasks at each node
and with some pointers to the positions of the string $w$. In particular, it does not contain any characters of $w$ and,
if all pointers are stored as relative values, it can be stored using $\Oh(|v|\log |v|)$ bits, i.e., $\Oh(|v|\frac{\log |v|}{\log n})$ words. For each $n_k$ the total length of selected substrings is $\Oh(n)$, and thus the space usage is
$\Oh(n\frac{\log n_k}{\log n})=\Oh(n 2^{-k})$, which sums up to $\Oh(n)$ over all lengths $n_k$.
The construction time is $\Oh(|v|\sqrt{\log |v|})$ for any substring (including alphabet renumbering),
and this sums up to $\Oh(n \sqrt{2^{-k}\log n})$ for each length, and $\Oh(n\sqrt{\log n})$ in total.
\mayqed
\end{proof}

\subsection*{Acknowledgement}
The authors gratefully thank Djamal Belazzougui, Travis Gagie, and Simon J. Puglisi who pointed out that wavelet suffix trees can be useful for BWT-related applications.

\bibliographystyle{plain}
\bibliography{main}

\appendix

\section{Constructing rank/select structures}\label{app:wt}

\begin{replemma}{lem:rank_select}
Given a bit string $B[1..N]$ packed in $\frac{N}{\log n}$ machine
words, we can extend it in $\Oh(\frac{N}{\log n})$ time with a rank/select data
structure occupying $o(\frac{N}{\log n})$ additional space, assuming $\tilde\Oh(\sqrt{n})$ time and space preprocessing shared by all instances
of the structure.
\end{replemma}
\begin{proof}
We focus on implementing $\rank_{1}$ and $\select_{1}$. By repeating the construction,
we also get $\rank_{0}$ and $\select_{0}$.

\paragraph{Rank.}
Recall the usual implementation of $\rank_{1}$. We first split $B$ into
superblocks of length $\log^{2}n$, and then split every superblock
into blocks of length $\frac{1}{2}\log n$. We store the cumulative rank
for each superblock, and the cumulative rank within the superblock for
each block. All ranks can be computed by reading whole blocks.
More precisely, in every step we extract the next $\frac{1}{2}\log n$
bits, i.e., take either the lower or the higher half of the next word
encoding the input, and compute the number of ones inside using
a shared table of size $\Oh(\sqrt{n})$. Hence we need $\Oh(\frac{N}{\log n})$-time preprocessing for rank queries.

\paragraph{Select.}
Now, recall the more involved implementation of $\select_{1}$. We split
$B$ into superblocks by choosing every $s$-th occurrence of {\bf 1}, where $s=\log n\log\log n$,
and store the starting position of every superblock together with a pointer to its description, using $\Oh(\frac{N}{s}\log n)=o(N)$ bits in total. This can be done by
processing the input word-by-word while maintaining the total number
of {\bf 1}s seen so far in the current superblock. After reading the next word, i.e.,
the next $\log n$ bits, we check if the next superblock should start
inside, and if so, compute its starting position. This can be easily done
if we can count {\bf 1}s inside the word and select the $k$-th one
in $\Oh(1)$ time, which can be preprocessed in a shared table of
size $\Oh(\sqrt{n})$ if we split every word into two parts. There are at most
$\frac{N}{s}$ superblocks, so the total construction time so far is
$\Oh(\frac{N}{\log n}+\frac{N}{s})=\Oh(\frac{N}{\log n})$. Then we have
two cases depending on the length $\ell$ of the superblock.

If $\ell > s^{2}$, we store the position of every {\bf 1} inside the superblock
explicitly. We generate the positions by sweeping the superblock word-by-word,
and extracting the {\bf 1}s one-by-one, which takes $\frac{\ell}{\log n}+s$ time.
Because there are at most $\frac{N}{s^{2}}$ such sparse superblocks, this is
$\Oh(\frac{N}{\log n}+\frac{N}{s})=\Oh(\frac{N}{\log n})$, and the
total number of used bits is $\Oh(\frac{N}{s^{2}}s\log n)=o(N)$.

If $\ell \leq s^{2}$, the standard solution is to recurse on the superblock,
i.e., repeat the above with $n$ replaced by $s^{2}$. We need to be more careful,
because otherwise we would have to pay $\Oh(1)$ time for roughly every $(\log\log n)^{2}$-th
occurrence of {\bf 1} in $B$, which would be too much.

We want to split every superblock into blocks
by choosing every $s'$-th occurrence of {\bf 1}, where $s'=(\log\log n)^{2}$, and storing
the starting position of every block relative to the starting position of its
superblock. Then, if a block is of length $\ell' > s'^{2}$, we store the position of
every {\bf 1} inside the block, again relative to the starting position of the superblock,
which takes $\Oh(\frac{N}{s'^{2}}s'\log(s^{2}))=o(N)$ bits in total. 
Otherwise, the block is so short that we can extract the $k$-th occurrence of {\bf 1} inside using
the shared table. The starting positions of the blocks are saved one after 
another using $\Oh(\frac{N}{s'}\log(s^{2}))=o(N)$ bits in total, so that we
can access the $i$-th block in $\Oh(1)$ time. With the starting
positions of every {\bf 1} in a sparse block the situation is slightly more complicated,
as not every block is sparse, and we cannot afford to store a separate pointer for every block.
The descriptions of all sparse blocks are also saved
one after another, but additionally we store for every block a single bit denoting
whether it is sparse or dense. These bits are concatenated together and enriched with
a $\rank_{1}$ structure, which allows us to compute in $\Oh(1)$ time where the
description of a given sparse block begins.

To partition a superblock into blocks efficiently, we must be able
to generate multiple blocks simultaneously. Assume that a superblock spans
a number of whole words (if not, we can shift its description paying
$\Oh(1)$ per word, which is fine). Then we can process these words
one-by-one after some preprocessing. First, consider the process of
creating the blocks when we read the bits one-by-one. We maintain the
description of the current block, which consists of its length $\ell' \le \ell$, the number
of {\bf 1}s seen so far, and their starting positions. The total size of the
description is $\Oh(\log \ell + s' \log \ell)$ bits. Additionally, we maintain
our current position inside the superblock, which takes another $\Oh(\log\ell)$
bits. After reading the next bit, we update the current position, and
either update the description of the current block, or start a new block.
In the latter case, we output the position of the first {\bf 1} inside
the current block, store one bit denoting whether $\ell' > s'^{2}$,
and if so, additionally save the positions
of all {\bf 1}s inside the current block. 

Now we want to accelerate processing the words describing a superblock.
Informally, we would like to read the whole next word, and generate all the corresponding
data in $\Oh(1)$ time. The difficulty is that the starting positions of the blocks, the
bits denoting whether a block is sparse or dense, and the descriptions of sparse blocks
are all stored in separate locations. 
To overcome this issue, we can view the whole process as an automaton with
one input tape, three separate output tapes, and a state described by
$\Oh(\log\ell+s'\log\ell)=\Oh(\polyloglog(n))$ bits. In every step,
the automaton reads the next bit from the input tape, updates its state,
and possibly writes into the output tapes. Because the situation is fully
deterministic, we can preprocess its behavior, i.e., for every initial state and
a sequence of $\frac{1}{2}\log n$ bits given in the input tape, we can precompute
the new state and the data written into the output tapes (observe that, by
construction, this is always at most $\log n$ bits). The cost of such
preprocessing is $\Oh(2^{\frac{1}{2}\log n+\polyloglog(n)})=\tilde\Oh(\sqrt{n})$.
The output buffers are implemented as packed lists, so appending at most
$\log n$ bits to any of them takes just $\Oh(1)$ time.
Therefore, we can process $\frac{1}{2}\log n$ bits in $\Oh(1)$ time, which accelerates
generating the blocks by the desired factor of $\log n$.
\mayqed
\end{proof}

\begin{replemma}{lm:rank/select construction 2}
Let $d\le \log^{\eps} n$ for $\eps<\frac13$.
Given a string $D[1..N]$ over the alphabet $[0,d-1]$ packed in $\frac{N\log d}{\log n}$ machine
words, we can extend it in $\Oh(\frac{N\log d}{\log n})$ time with a generalized rank/select
structure occupying additional $o(\frac{N}{\log n})$ space, assuming
an $\tilde\Oh(\sqrt{n})$ time and space preprocessing shared by all instances
of the structure.
\end{replemma}
\begin{proof}
The construction is similar to the one from Lemma~\ref{lem:rank_select},
but requires careful adjustment of the parameters. The main difficulty is that
we cannot afford to consider every $c\in [0,d-1]$ separately as in the binary
case.

\paragraph{Rank.}
We split $D$ into superblocks of length $d \log^{2} n$, and then split every superblock into
blocks of length $\frac{1}{3}\frac{\log n}{\log d}$. For every superblock, we store the
cumulative generalized rank of every character, i.e., for every character $c$ we store the number of positions
where characters $c'\le c$ occur in the prefix of the string up to the beginning of the superblock.
This takes $\Oh(d\log n)$ bits per superblock, which is $\Oh(\frac{N}{\log n})=o(N)$ in total.

For every block, we store the cumulative generalized rank of every character
within the superblock, i.e., for every
character $c$ we store the number of positions where characters $c'\le c$ occur in the prefix
of the superblock up to the beginning of the block. 
This takes $\Oh(d\log(d \log^{2} n))=\Oh(\log^{\eps}n\log\log n)$ bits per block,
which is $\Oh(N/\frac{\log n}{\log d}\log^{\eps}n\log\log n)=o(N)$ in total.

Finally, we store a shared table, which given
a string of $\frac{1}{3}\frac{\log n}{\log d}$ characters packed into
a single machine word of length $\frac{1}{3}\log n$, a character $c$, and a number
$i$, returns the number of positions where characters $c'\le c$ occur in the prefix of length $i$.
Both the size and the construction time of the table is $\tilde\Oh(2^{\frac{1}{3}\log n})$. 
It is clear that using the table, together with the cumulative ranks
within the blocks and superblocks, we can compute any generalized rank in $\Oh(1)$ time.

We proceed with a construction algorithm.
All ranks can be computed by
reading whole blocks. More precisely, the cumulative generalized ranks of the
next block can be computed by taking the cumulative generalized ranks of
the previous block (if it belongs to the same superblock) stored in one word
of length $d\log(d\log^{2}n)$, and updating it with
the next $\frac{1}{3}\frac{\log n}{\log d}$ characters stored in one word of
length $\frac{1}{3}\log n$. Such updates can be preprocessed for every possible
input in $\tilde\Oh(2^{\log^{\eps}n+\frac{1}{3}\log n})$ time and space. 
Then each block can be handled in $\Oh(1)$ time, which gives $\Oh(\frac{N\log d}{\log n})$ in total.

The cumulative generalized ranks of the next superblock can be computed
by taking the cumulative generalized ranks of the previous superblock,
and updating it with the cumulative generalized ranks of the last block in that previous superblock.
This can be done by simply iterating over the whole alphabet and updating each
generalized rank separately in $\Oh(1)$ time, which is $\Oh(\frac{N}{d\log^{2}n}d)=o(\frac{N\log d}{\log n})$ in total.

\paragraph{Select.} We store a separate structure for every $c\in[0,d-1]$. Even though the
structures must be constructed together to guarantee the promised complexity,
first we describe a single such structure. Nevertheless, when stating
the total space usage, we mean the sum over all $c\in [0,d-1]$.

We split $D$ into superblocks by choosing every $s$-th occurrence of $c$ in $D$,
where $s=d\log^{2} n$, and store the starting position of every superblock.
This takes $\Oh((d+\frac{N}{s})\log n)=o(N)$ bits in total. As in the binary
case, then we proceed differently depending on the length $\ell$ of the superblock.

If $\ell > s^{2}$, we store the position of every occurrence of $c$ inside the
superblock explicitly. This takes $\Oh(d\frac{N}{s^{2}}s\log n)=o(N)$ bits.

If $\ell \leq s^{2}$, we split the superblock into smaller blocks by choosing
every $s'$-th occurrence of $c$ inside, where $s'=d (\log \log n)^2$. We write
down the starting position of every block relative to the starting position of the
superblock, which takes $\Oh((d+\frac{N}{s'})\log s)=\Oh(\frac{N}{d\log\log n})=o(\frac{N}{d})$ bits in total. Then, if the block
is of length $\ell ' > d \cdot s'^{2}$, we store the position of every occurrence of $c$
inside, again relative to the starting position of the superblock, which takes
$\Oh(d\frac{N}{d\cdot s'^{2}}s'\log s)=\Oh(\frac{N}{d\log\log n})=o(\frac{N}{d})$ bits in total. Otherwise, the block spans at most
$\Oh(\log^{3\eps}n (\log \log n)^{4})=o(\log n)$ bits in $D$,
hence the $k$-th occurrence of $c$ there can be extracted in $\Oh(1)$ time using a shared
table of size $\Oh(\sqrt{n})$. Descriptions of the sparse blocks are stored
one after another as in the binary case.

This completes the description of the structure for a single $c\in [0,d-1]$. To access the
relevant structure when answering a select query, we additionally store the pointers
to the structures in an array of size $\Oh(d)$. Now we will argue that all structures
can be constructed in $\Oh(\frac{N\log d}{\log n})$ time, but first
let us recall what data is stored for a single $c\in [0,d-1]$:
\begin{enumerate}[(1)]\compact
\item an array storing the starting positions of all superblocks and pointers to their descriptions,
\item for every sparse superblock, a list of positions of every occurrence of $c$ inside,
\item for every dense superblock, a bitvector with the $i$-th bit denoting if the $i$-th block inside
is sparse,
\item an array storing the starting positions of all blocks relative to the starting position of their superblock,
\item for every sparse block, a list of positions of every occurrence of $c$ inside relative to the starting
position of its superblock.
\end{enumerate}
First, we compute the starting positions of all superblocks by scanning the input while maintaining
the number of occurrences of every $c\in [0,d-1]$ in its current superblock. These
counters take $\Oh(d\log s)=o(\log n)$ bits together, and are packed in a single machine word.
We read the input word-by-word and update the counters, which can be done in $\Oh(1)$ time
after an appropriate preprocessing. Whenever we notice that a counter exceeds $s$, we output
a new block. This requires just $\Oh(1)$ additional time per a superblock, assuming an
appropriate preprocessing. The total time complexity is hence $\Oh(\frac{N\log d}{\log n})$ so far.
Observe that we now know which superblock is sparse.

To compute the occurrences in sparse superblocks, we perform another word-by-word scan of the input. During the scan, we keep track of the current superblock
for every $c\in [0,d-1]$, and if it is sparse, we extract and copy the starting positions of all
occurrences of $c$. This requires just $\Oh(1)$ additional time per an occurrence of $c$ in a sparse
superblock, which sums up to $\Oh(\frac{N}{\log^{2} n})$.

The remaining part is to split all dense superblocks into blocks. Again, we view the the process as
an automaton with one input tape, three output tapes corresponding to different parts of the structure, and
a state consisting of $\Oh(\log\ell+s'\log \ell)$ bits. Here, we must be more careful than
in the binary case: because we will be reading the whole input word-by-word, not just a single
dense superblock, it might happen that $\ell$ is large. Therefore, before feeding the input into the
automaton, we mask out all occurrences of $c$ within sparse superblocks, which can be done
in $\Oh(1)$ time per an input word after an appropriate preprocessing, if we keep track of the current
superblock for every $c\in [0,d-1]$ while scanning through the input. Then we can combine all $d$ automata
into a larger automaton equipped with $3d$ separate output tapes and a state consisting of $o(\log n)$ bits. We preprocess
the behavior of the large automaton after reading a sequence of $\frac{1}{2}\log n$ bits given
in the input tape for every initial state. The result is the new state, and the data written into
each of the $3d$ output tapes. Now we again must be more careful: even though the output tapes are
implemented as packed lists, and we need just $\Oh(1)$ time to append the preprocessed data to any them,
we cannot afford to touch every output tape
every time we read the next $\frac{1}{2}\log n$ bits from the input.
Therefore, every output tape is equipped with an output buffer consisting of $\frac{1}{5}\frac{\log n}{d}$ bits,
and all these output buffers are stored in a single machine word consisting of $\frac{1}{5}\log n$ bits.
Then we modify the preprocessing: for every initial state, a sequence of $\frac{1}{5}\log n$ input bits,
and the initial content of the output buffers, we compute the new state, the new content of the output
buffers, and zero or more full chunks of $\frac{1}{5}\frac{\log n}{d}$ bits that should be written into the respective
output tapes. Such preprocessing is still $\Oh(\sqrt{n})$, and now the additional time is just $\Oh(1)$
per each such chunk. Because we have chosen the parameters so that the size of the additional data in bits
is $o(\frac{N}{d})$, the total time complexity is $\Oh(\frac{N\log d}{\log n})$ as claimed.
\mayqed
\end{proof}

\end{document}